\documentclass[journal,10pt,doublecolumn]{IEEEtran}
\usepackage{amsmath}
\usepackage{graphicx}
\usepackage{cite}
\usepackage{amsfonts}
\usepackage{amsthm}
\usepackage{color}
\usepackage{algorithm}
\usepackage{algorithmic}
\usepackage{amssymb}

\theoremstyle{plain}  \newtheorem{lemma}{Lemma}
\theoremstyle{plain}  
\theoremstyle{plain}

\begin{document}



\title{Wireless Powered Cooperative Jamming for Secure OFDM System}

\author{Guangchi~Zhang, Jie~Xu, Qingqing~Wu, Miao~Cui, Xueyi~Li, and Fan~Lin
\thanks{G. Zhang, J. Xu, M. Cui, and X. Li are with the School of Information Engineering, Guangdong University of Technology, Guangzhou, China (e-mail: gczhang@gdut.edu.cn, jiexu@gdut.edu.cn, single450@163.com, leexueyi@gdut.edu.cn). J. Xu is the corresponding author.  

Q. Wu is with the Department of Electrical and Computer Engineering, National University of Singapore (e-mail: elewuqq@nus.edu.sg).
	
F. Lin is with Guangzhou GCI Science \& Technology Co., Ltd., Guangzhou, China (e-mail: linfan@chinagci.com). }}

\date{}

\markboth{}%
{}

\maketitle

\begin{abstract}
This paper studies the secrecy communication in an orthogonal frequency division multiplexing (OFDM) system, where a source sends confidential information to a destination in the presence of a potential eavesdropper. We employ wireless powered cooperative jamming to improve the secrecy rate of this system with the assistance of a cooperative jammer, which works in the harvest-then-jam protocol over two time-slots. In the first slot, the source sends dedicated energy signals to power the jammer; in the second slot, the jammer uses the harvested energy to jam the eavesdropper, in order to protect the simultaneous secrecy communication from the source to the destination. In particular, we consider two types of receivers at the destination, namely Type-I and Type-II receivers, which do not have and have the capability of canceling the (a-priori known) jamming signals, respectively. For both types of receivers, we maximize the secrecy rate at the destination by jointly optimizing the transmit power allocation at the source and the jammer over sub-carriers, as well as the time allocation between the two time-slots. First, we present the globally optimal solution to this problem via the Lagrange dual method, which, however, is of high implementation complexity. Next, to balance tradeoff between the algorithm complexity and performance, we propose alternative low-complexity solutions based on minorization maximization and heuristic successive optimization, respectively. Simulation results show that the proposed approaches significantly improve the secrecy rate, as compared to benchmark schemes without joint power and time allocation.
\end{abstract}


\section{Introduction}
With recent technical advancements in Internet of things (IoT), future wireless networks are envisioned to incorporate billions of low-power wireless devices to enable various industrial and commercial applications \cite{Rico2016}. How to ensure the confidentiality of these devices' wireless communication against illegitimate eavesdropping attacks is becoming an increasingly important task for cyber-physical security. However, this task is particularly challenging, as conventional key-based cryptographic techniques are difficult to be implemented due to the broadcast nature of wireless communications. To overcome this issue, physical layer security has emerged as a viable anti-eavesdropping solution at the physical layer \cite{Khisti2010, Mukherjee2014, Li2011}. The key design objective in physical-layer security is to maximize the so-called secrecy rate, which is defined as the communication rate of a wireless channel, provided that eavesdroppers cannot overhear any information from this channel. 

In the literature, there have been various approaches proposed to improve the secrecy rate. For example, one widely adopted approach is based on the idea of artificial noise (AN) (see, e.g., \cite{Goel2008, Khandaker2015}). In this approach, wireless transmitters send a combined version of both confidential information signals and AN, where the AN acts as jamming signals to interfere with eavesdroppers, thus avoiding the information leakage. Another celebrated approach is called cooperative jamming (see, e.g., \cite{Huang2011, QiangLi2015, Luo2013}), where external network nodes cooperatively send jamming signals to disrupt the eavesdropping, thus helping protect the confidential information communication. As compared to the AN-based approach, cooperative jamming is able to further improve the secrecy rate by exploiting the cooperation diversity among different nodes. Cooperative jamming is also expected to have more abundant applications in the IoT era, where massive low-power wireless devices can cooperate in jamming to improve the network security. For instance, some idle devices in wireless networks can act as cooperative jammers to help ensuring the secrecy communication of other actively communicating devices. 

Nevertheless, the practical implementation of cooperative jamming in IoT networks is hindered by the low-power nature of wireless devices, since cooperative jamming will consume energy on these devices and thus they may prefer keeping idle to save energy instead of involving in the cooperation. To overcome this issue, a new efficient method, namely wireless powered cooperative jamming, has been proposed in \cite{Liu2016, YingBi2016, Xing2015, Xing2016} motivated by the recent success of wireless information and power transfer via radio frequency (RF) signals \cite{ZhangRui2013, Zhu2016, Xu2014, Xu2014b, XiaMinghua2015, Ng2014, Bi2015, Zeng2015, Krikidis2014, Wu2016, Wu2017, Xu2016, Luo2015}.\footnote{It is worth noting that in addition to the far-field RF-based wireless power transfer, magnetic induction is a widely used near-field wireless power transfer technique for charging electronic devices \cite{Krikidis2014,Wu2017}. However, the magnetic induction has a limited operating range of less than one meter in general, which is much shorter than that of the RF-based wireless power transfer in the order of several meters. Therefore, RF-based wireless power transfer is expected to have more abundant applications to charge low-power IoT devices in a wide range, and thus is considered here in the wireless powered cooperative jamming systems.} In this method, the cooperative jamming is powered by the wireless energy transferred from external wireless transmitters, and does not require cooperative jammers to consume their own energy. Therefore, wireless powered cooperative jamming is a promising solution to inspire low-power IoT devices to cooperate in the jamming. In \cite{Liu2016, YingBi2016}, wireless powered cooperative jamming was employed to secure a point-to-point communication system in the presence of an eavesdropper, where a cooperative jammer operates in an accumulate-and-jam protocol by first harvesting the wireless energy and storing in the battery over multiple blocks and then using the accumulated energy for cooperative jamming. The long-term secrecy performance is optimized by adjusting jamming parameters while taking into account the channel and battery dynamics over time. In \cite{Xing2015, Xing2016}, wireless powered cooperative jamming was used in a secrecy two-way relaying communication system, where an eavesdropper aims to intercept the communicated information at the second hop, and more than one cooperative jammers operate in a harvest-then-jam protocol for cooperative jamming: in the first slot, the jammers harvest the wireless energy from the source, while in the second slot, they use the harvested energy to cooperatively jam the eavesdroppers. As the harvested energy is immediately used in the following slot, the harvest-then-jam protocol does not require large-capacity energy storages nor sophisticated energy management at cooperative jammers. For this reason, it is generally much easier to be implemented in practice than the accumulate-and-jam protocol. 

In this paper, we consider wireless powered cooperative jamming to secure a point-to-point communication system from a source to a destination with the presence of a potential eavesdropper. Different from prior works considering single-carrier systems, we focus on the multi-carrier orthogonal frequency division multiplexing (OFDM) system, which offers the following advantages. First, note that the wireless transmission must meet the transmit power spectrum density constraints imposed by regulatory authorities. In this case, the transferred power over a narrow-band system is often limited. By contrast, using OFDM over a wideband wireless power transfer system and exploiting the channel diversity over frequency can help deliver more power to intended receivers. On the other hand, as OFDM has been widely adopted in major existing and future wireless communication networks, using it here can also help better integrate wireless power transfer and wireless communication for future wireless networks (see, e.g., \cite{Wu2017, Zhou2014, Renna2012, Zhang2016, Wang2013} and references therein). The cooperative jammer works in a harvest-then-jam protocol to help the secrecy communication by dividing each transmission block into two time-slots: in the first slot, the source sends dedicated energy signals to power the jammer; while in the second slot, the jammer uses the harvested energy to interfere with the eavesdropper to protect the confidential information transmission. 

In general, there exists a tradeoff in the time allocation between the two slots to optimize the performance of secrecy communication, i.e., while a longer WPT time in the first slot can transfer more energy to increase the jamming power for better confusing the eavesdropper, it can also reduce the efficient wireless information transmission (WIT) time in the second slot for delivering confidential data. Therefore, in order to improve the secrecy rate at the destination by maximally exploring the benefit of wireless power cooperative jamming, it is important to jointly design the time allocation, together with the transmit power allocation at the source and the jammer over sub-carriers, by taking into account the energy harvesting constraint at the jammer. We maximize the secrecy rate via joint time and power allocation by particularly considering two types of receivers at the destination, namely Type-I and Type-II receivers, which do not have and have the capability of canceling the (a-priori known) jamming signals, respectively (see Section II for the details). Under both receiver types, however, the two joint time and power allocation problems are non-convex and usually difficult to be solved. To tackle such challenges, we propose to recast each problem into a two-layer form, in which the outer layer corresponds to a single-variable time allocation problem and the inner layer is a sub-carrier transmit power allocation problem under given time allocation. The outer layer time allocation problem is solved via a one-dimension search. As for the inner-layer power allocation problem, we first present the globally optimal solution via the Lagrange dual method, which, however, is of high implementation complexity. Next, to balance the tradeoff between the implementation complexity and the performance, we further develop two suboptimal solutions based on minorization maximization and heuristic successive optimization, respectively. Simulation results show that the proposed approaches achieve significantly higher secrecy rate than benchmark schemes without joint time and power allocation, and the minorization maximization based suboptimal solution achieves a near optimal performance as compared to the optimal solution. 

It is worth noting that in the literature, there have been several existing works \cite{Renna2012, Zhang2016, Wang2013} investigating the physical layer security over OFDM systems. For example, the secrecy rate of OFDM systems was investigated in \cite{Renna2012} under a Rayleigh fading channel setup without using AN or cooperative jamming. In \cite{Zhang2016} and \cite{Wang2013}, the AN-based approach and cooperative jamming were considered to improve the secrecy rate of OFDM systems, respectively. Different from these prior studies, in this paper the cooperative jamming is powered by WPT, and thus requires a more sophisticated design with joint time and power allocation for both WPT and jamming. This is new and has not been addressed. 

The remainder of the paper is organized as follows. Section II presents the system model and problem formulation. Sections III and IV propose three efficient approaches to obtain solutions to the two joint time and power allocation problems with Type-I and Type-II destination receivers, respectively. Section V presents simulation results to validate the performance of our proposed joint design as compared to other benchmark schemes. Finally, Section VI concludes this paper.

\begin{figure}
	\centering
	\includegraphics[width=0.48\textwidth]{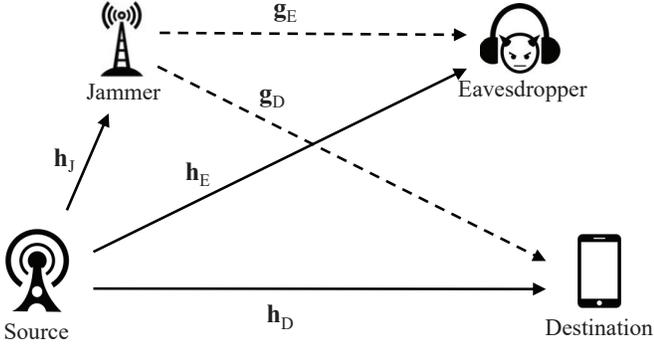}
	\caption{A secure OFDM communication system with wireless powered cooperative jamming, where a source sends confidential information to a destination in the presence of a potential eavesdropper, and a cooperative jammer uses the harvested wireless energy from the source to jam the eavesdropper against its eavesdropping.}
	\label{fig1}
\end{figure}

\section{System Model and Problem Formulation}
\subsection{System Model}
As shown in Fig. \ref{fig1}, we consider secrecy communication in an OFDM system with a source communicating with a destination in the presence of a potential eavesdropper. We employ wireless powered cooperative jamming to secure this system, where a cooperative jammer uses the transferred energy from the source to help jam the eavesdropper against its eavesdropping. Suppose that the OFDM system consists of a total of $N$ orthogonal sub-carriers, and denote the set of sub-carriers as $\mathcal{N} \triangleq \{ 1, 2, \ldots, N \}$. We consider a block-based quasi-static channel model by assuming that the wireless channels remain constant over each transmission block and may change from one block to another. We focus on one particular block with a length of $T$, and denote $\mathbf{h}_{\text{J}} \triangleq [ h_{\text{J},1} , \ldots, h_{\text{J},N}]^\dagger$, $\mathbf{h}_{\text{D}} \triangleq [ h_{\text{D},1} , \ldots, h_{\text{D},N}]^\dagger$, $\mathbf{h}_{\text{E}} \triangleq [ h_{\text{E},1} , \ldots, h_{\text{E},N}]^\dagger$, $\mathbf{g}_{\text{D}} \triangleq [ g_{\text{D},1} , \ldots, g_{\text{D},N}]^\dagger$, $\mathbf{g}_{\text{E}} \triangleq [ g_{\text{E},1} , \ldots, g_{\text{E},N}]^\dagger$ as the vectors collecting the channel coefficients of all the $N$ sub-carriers from the source to the jammer, from the source to the destination, from the source to the eavesdropper, from the jammer to the destination, from the jammer to the eavesdropper, respectively. Here, the superscript $\dagger$ denotes the transpose operation. It is assumed that the source, destination, and the cooperative jammer perfectly know the global channel state information (CSI) $\mathbf{h}_{\text{J}}$, $\mathbf{h}_{\text{D}}$, $\mathbf{h}_{\text{E}}$, $\mathbf{g}_{\text{D}}$, and $\mathbf{g}_{\text{E}}$ in order to obtain the performance upper bound of the wireless powered cooperative jamming system. Specifically, the CSI $\mathbf{h}_{\text{J}}$, $\mathbf{h}_{\text{D}}$, and $\mathbf{g}_{\text{D}}$ associated with these users can be obtained via efficient channel estimation and feedback among them, while $\mathbf{g}_{\text{E}}$ and $\mathbf{h}_{\text{E}}$ can be obtained by monitoring the possible transmission activities of the eavesdropper, as commonly assumed in the physical-layer security literature \cite{Khisti2010, Li2011, Goel2008, Khandaker2015, Huang2011}. Note that in practice the CSI acquisition may consume additional energy at the cooperative jammer, and the obtained CSI may not be perfect due to channel estimation and feedback errors. However, how to address these issues in practice is left for future work.

We consider a harvest-then-jam protocol for the cooperative jammer by dividing each transmission block into two time-slots with lengths $\alpha_1 T$ and $\alpha_2 T$, respectively, where $\alpha_1$ and $\alpha_2$ denote the portions of the two time-slots with 
\begin{equation}  \label{EquTimeCon}
\alpha_1 + \alpha_2 = 1, \; \; 0 \le \alpha_1 \le 1, \; 0 \le \alpha_2 \le 1.
\end{equation} 
In the first time-slot, the source sends wireless energy to power the cooperative jammer; while in the second time-slot, the source transmits confidential information to the destination and simultaneously the jammer uses the harvested energy in the first time-slot to cooperate in jamming the eavesdropper against its eavesdropping. The detailed operation in the two slots is presented in the following, respectively. 

First, consider the WPT from the source to the jammer in the first time-slot. Over each sub-carrier $n$, let $s_{\text{PT},n}$ denote the energy signal transmitted by the source, which is assumed to be a random variable with variance $\mathbb{E}( | s_{\text{PT},n} |^2 ) = p_{\text{PT},n} $. Here, $p_{\text{PT},n}$ denotes the transmit power for WPT at the source over the sub-carrier $n$, and $\mathbb{E}(\cdot)$ denotes the statistic expectation. The harvested energy by the jammer is
\begin{equation}  \label{EquEEH}
E_{\text{EH}} = \alpha_1 T \eta \sum_{n=1}^{N} p_{\text{PT},n} | h_{\text{J},n} |^2,
\end{equation}
where $0 < \eta \leq 1$ denotes the energy harvesting efficiency at the jammer. Note that similarly as in \cite{ZhangRui2013,Zhu2016,Xu2014,Xu2014b,XiaMinghua2015,Ng2014,Bi2015,Zeng2015}, we adopt a linear energy harvesting model in \eqref{EquEEH} by considering the harvested power at the jammer lies in the linear regime of the energy harvester. In the literature, there have been various works \cite{Boshkovska2015,Moghadam2017,Zeng2017,Clerckx2016} investigating the wireless power transfer by considering the non-linearity of the energy harvester, while how to extend the wireless powered cooperative jamming into such a scenario is left for future work.

Next, consider the cooperative jamming in the second time-slot. Over the sub-carrier $n$, let $s_{\text{IT},n}$ and $s_{\text{J},n}$ denote the confidential information signal transmitted by the source and the jamming signal transmitted by the jammer, respectively. The received signals by the destination and the eavesdropper over the sub-carrier $n$ are respectively denoted as
\begin{equation}  \label{EquYD}
y_{ \text{D}, n } = h_{\text{D},n} s_{\text{IT},n} + g_{\text{D},n} s_{\text{J},n} + n_{\text{D},n},
\end{equation}
\begin{equation}   \label{EquYE}
y_{ \text{E}, n } = h_{\text{E},n} s_{\text{IT},n} + g_{\text{E},n} s_{\text{J},n} + n_{\text{E},n},
\end{equation}
where $n_{\text{D},n}$ and $n_{\text{E},n}$ denote the Gaussian noise at the receivers of the destination and the eavesdropper with mean zero and variances $\sigma_{\text{D}}^2$ and $\sigma_{\text{E}}^2$, respectively. Assume that Gaussian signaling is employed for both $s_{\text{IT},n}$ and $s_{\text{J},n}$, which are thus cyclic symmetric complex Gaussian (CSCG) random variables with mean zero and variances $\mathbb{E}( | s_{\text{IT},n} |^2 ) = p_{\text{IT},n}$ and $\mathbb{E}( | s_{\text{J},n} |^2 ) = p_{\text{J},n}$, with $p_{\text{IT},n}$ and $p_{\text{J},n}$ denoting the transmit power of the source and the jamming power of the jammer over the sub-carrier $n$, respectively. Let $P_S$ denote the maximum transmit sum power of the source over all sub-carriers, and $P_{\text{S,peak}}$ denote the peak transmit power of the source over each sub-carrier. Then we have 
\begin{subequations}   \label{EquSMaxPowCon}
\begin{align}
& \alpha_1 \sum_{n=1}^{N} p_{\text{PT},n} + \alpha_2 \sum_{n=1}^N p_{\text{IT},n} \leq P_S,  \\
& 0 \leq p_{\text{PT},n}, p_{\text{IT},n}  \leq P_{\text{S,peak}},  \; n \in \mathcal{N}.
\end{align}
\end{subequations}
As for the jammer, as it uses the harvested wireless energy $E_{\text{EH}}$ in \eqref{EquEEH} in the first time-slot to supply the cooperative jamming in the second time-slot, it is subject to the energy harvesting constraint: the total energy used for jamming in the second time-slot cannot exceed $E_{\text{EH}}$, i.e.,
\begin{subequations}    \label{EquJMaxPowCon}
\begin{align}
& \alpha_2 T \sum_{n=1}^N p_{\text{J},n} \leq E_{\text{EH}} = \alpha_1 T \eta \sum_{n=1}^{N} p_{\text{PT},n} | h_{\text{J},n} |^2, \\
& 0 \leq p_{\text{J},n}  \leq P_{\text{J,peak}}  , \; n \in \mathcal{N},
\end{align}
\end{subequations}
where $P_{\text{J,peak}}$ denotes the peak transmit power of the jammer over each sub-carrier.

In particular, we consider two types of receivers at the destination, namely Type-I and Type-II receivers \cite{Xu2014}, which do not have and have the capability of canceling the jamming signals $s_{\text{J},n}$'s from the jammer, respectively. In order for a Type-II receiver to successfully cancel the jamming signals, such signals should be securely shared between the jammer and the destination before the cooperative jamming \cite{Zhang2016, Xu2014, Liu2014, Xing2016b}. This can be practically implemented as follows \cite{Xing2016b}. First, the same jamming signal generators and seed tables are pre-stored at both the jammer and destination (but not available at the eavesdropper). Next, before each transmission phase, one seed is randomly chosen from the seed table and the index of this seed is shared between the jammer and destination. In particular, the two-step phase-shift modulation-based method in \cite{Xing2016b} can be applied for the seed index sharing as follows. In the first step, the destination sends a pilot signal for the jammer to estimate the channel phase between the destination and jammer. In the second step, the jammer randomly chooses a seed index, and modulates it over the phase of the transmitted signal after pre-compensating the channel phase that it estimated in the previous step. The destination is able to decode the seed index sent by the jammer from the phases of the received signal. Since the length of this seed index sharing procedure is very short and the channel phase between the destination and jammer is different from that between the destination/jammer and the eavesdropper, the eavesdropper does not know the channel phase between the destination and jammer, and thus is not able to decode the signal containing the seed index in such a short time period. For Type-I and Type-II receivers, the secrecy rates of the secure OFDM system over the $N$ sub-carriers are respectively given by
\begin{equation}   \label{EquRate1}
R_{\text{sec}}^{(\text{I})} = \sum_{n=1}^N  \left[ R_{\text{SD},n}^{(\text{I})} - R_{\text{SE},n} \right]^{+},
\end{equation}
\begin{equation}  \label{EquRate2}
R_{\text{sec}}^{(\text{II})} = \sum_{n=1}^N  \left[ R_{\text{SD},n}^{(\text{II})} - R_{\text{SE},n} \right]^{+},
\end{equation}
where $[x]^{+} \triangleq \max(x,0)$. Here, $R_{\text{SD},n}^{(\text{I})}$ and $R_{\text{SD},n}^{(\text{II})}$ are the achievable rates over the sub-carrier $n$ from the source to the destination for Type-I and Type-II receivers, respectively, and $R_{\text{SE},n}$ denotes the achievable rate from the source to the eavesdropper over the sub-carrier $n$, given by
\begin{equation}
R_{\text{SD},n}^{(\text{I})} = \alpha_2  \log_2 \left( 1 + \frac{ p_{\text{IT},n} | h_{\text{D},n} |^2 } { p_{\text{J},n} | g_{\text{D},n} |^2 + \sigma_{\text{D}}^2  }  \right),
\end{equation}
\begin{equation}
R_{\text{SD},n}^{(\text{II})} = \alpha_2  \log_2 \left( 1 + \frac{ p_{\text{IT},n} | h_{\text{D},n} |^2 } { \sigma_{\text{D}}^2  }  \right),
\end{equation}
\begin{equation}
R_{\text{SE},n} = \alpha_2  \log_2 \left( 1 + \frac{ p_{\text{IT},n} | h_{\text{E},n} |^2 } { p_{\text{J},n} | g_{\text{E},n} |^2 + \sigma_{\text{E}}^2  }  \right).
\end{equation}

\subsection{Problem Formulation}
Our objective is to maximize the secrecy rates $R_{\text{sec}}^{(\text{I})}$ in \eqref{EquRate1} and $R_{\text{sec}}^{(\text{II})}$ in \eqref{EquRate2} for both types of destination receivers, subject to the transmit power constraint in \eqref{EquSMaxPowCon} at the source, the energy harvesting constraint in \eqref{EquJMaxPowCon} at the jammer, and the time constraint in \eqref{EquTimeCon}. The decision variables include the transmit power allocation $p_{\text{PT},n}$'s (for WPT) and $p_{\text{IT},n}$'s (for WIT) at the source, and the jamming power allocation $p_{\text{J},n}$'s at the jammer, as well as the time allocation $\alpha_1$ and $\alpha_2$. For Type-I receiver, we mathematically formulate the secrecy rate maximization problem as 
\begin{align}  
(\text{P}1): \max_{ \substack{ \alpha_1, \alpha_2, \\ \mathbf{p}_{\text{PT}}, \mathbf{p}_{\text{IT}}, \mathbf{p}_{\text{J}} } } & \alpha_2 \sum_{n=1}^N \bigg[  \log_2 \left( 1 + \frac{ p_{\text{IT},n} | h_{\text{D},n} |^2 } {  p_{\text{J},n} | g_{\text{D},n} |^2 + \sigma_{\text{D}}^2 }  \right)  \nonumber \\
& -  \log_2 \left( 1 + \frac{ p_{\text{IT},n} | h_{\text{E},n} |^2 } {  p_{\text{J},n} | g_{\text{E},n} |^2 + \sigma_{\text{E}}^2  }  \right)  \bigg]     \label{EquOrigProb}  \\
\text{s.t.}  \quad  & \eqref{EquTimeCon}, \; \eqref{EquSMaxPowCon}, \; \eqref{EquJMaxPowCon}, \nonumber
\end{align}
where $\mathbf{p}_{\text{PT}} \triangleq [ p_{\text{PT},1}, \ldots,p_{\text{PT},N} ]^\dagger$, $\mathbf{p}_{\text{IT}} \triangleq [ p_{\text{IT},1}, \ldots,p_{\text{IT},N} ]^\dagger$, and $\mathbf{p}_{\text{J}} \triangleq [ p_{\text{J},1},\ldots,p_{\text{J},N} ]^\dagger$. Note that in the objective function of problem (P1) we have omitted the positive operation $[\cdot]^+$, which is due to the fact that the optimal value of each summation term of the objective of problem (P1), i.e. $R_{\text{SD},n}^{(\text{I})} - R_{\text{SE},n}$, must be non-negative, and thus the problems with and without the positive operation have the same optimal value and the same optimal solution.\footnote{This fact can be proved by contradiction. If $R_{\text{SD},n}^{(\text{I})} - R_{\text{SE},n} < 0$, we can increase its value to zero by setting $p_{\text{IT},n}=0$ without violating the constraints.}

Similarly, for Type-II receiver, the secrecy rate maximization problem is formulated as
\begin{align}   
(\text{P}2): \max_{ \substack{ \alpha_1, \alpha_2, \\ \mathbf{p}_{\text{PT}}, \mathbf{p}_{\text{IT}}, \mathbf{p}_{\text{J}} } }  & \alpha_2 \sum_{n=1}^N \bigg[  \log_2 \left( 1 + \frac{ p_{\text{IT},n} | h_{\text{D},n} |^2 } { \sigma_{\text{D}}^2 }  \right) \nonumber \\
& -  \log_2 \left( 1 + \frac{ p_{\text{IT},n} | h_{\text{E},n} |^2 } {  p_{\text{J},n} | g_{\text{E},n} |^2 + \sigma_{\text{E}}^2  }  \right)  \bigg]   \label{EquOrigProbType2}  \\
\text{s.t.}  \quad  & \eqref{EquTimeCon}, \; \eqref{EquSMaxPowCon}, \; \eqref{EquJMaxPowCon}.  \nonumber
\end{align}

Note that problems (P1) and (P2) are non-convex as their objective functions are non-concave. As a result, they are difficult to solve in general. In the following two sections, we tackle such difficulties for (P1) and (P2), respectively.

\section{Solution to Problem (P1) with Type-I Destination Receiver}
First, consider problem (P1) with Type-I destination receiver. We solve this problem by formulating it in a nested form:  
\begin{equation}  \label{EquOutLayProb}
\max_{\alpha_2} \;  \alpha_2 \mathcal{R}^{(\text{I})}(\alpha_2) , \; \; \text{s.t.} \;  0 \leq \alpha_2 \leq 1,
\end{equation}
where
\begin{subequations}   \label{EquInnLayProb}
	\begin{align}
	\mathcal{R}^{(\text{I})}(\alpha_2) = & \max_{ \mathbf{p}_{\text{PT}}, \mathbf{p}_{\text{IT}}, \mathbf{p}_{\text{J}} }    \sum_{n=1}^N \bigg[  \log_2 \left( 1 + \frac{ p_{\text{IT},n} | h_{\text{D},n} |^2 } {  p_{\text{J},n} | g_{\text{D},n} |^2 + \sigma_{\text{D}}^2 }  \right) \nonumber \\
	&\;\;\;\; \quad -  \log_2 \left( 1 + \frac{ p_{\text{IT},n} | h_{\text{E},n} |^2 } {  p_{\text{J},n} | g_{\text{E},n} |^2 + \sigma_{\text{E}}^2 }  \right)  \bigg]   \label{EquInnObj}  \\
	&\text{s.t.} \;      (1- \alpha_2)  \sum_{n=1}^{N} p_{\text{PT},n} + \alpha_2 \sum_{n=1}^N p_{\text{IT},n} \leq P_S , \label{EquInnCon1} \\
	&\quad  \; \;  0 \leq p_{\text{PT},n}, p_{\text{IT},n}  \leq P_{\text{S,peak}} ,\; n \in \mathcal{N},  \label{EquInnCon2}  \\
	&\quad \; \;  \alpha_2 \sum_{n=1}^N p_{\text{J},n} \leq (1-\alpha_2) \eta \sum_{n=1}^{N} p_{\text{PT},n} | h_{\text{J},n} |^2 , \label{EquInnCon3} \\
	&\quad  \; \;  0 \leq p_{\text{J},n}  \leq P_{\text{J,peak}}  , \; n \in \mathcal{N}.  \label{EquInnCon4}
	\end{align}
\end{subequations}
Here, the outer layer problem \eqref{EquOutLayProb} corresponds to the time allocation via optimizing $\alpha_2$, while the inner layer problem \eqref{EquInnLayProb} corresponds to the joint power allocation optimization under given time allocation. We solve problem (P1) by first solving \eqref{EquInnLayProb} under any given $\alpha_2 \in [0,1]$, and then adopting a one-dimensional search over the interval $[0,1]$ to find the optimal $\alpha_2$ to solve \eqref{EquOutLayProb}. In the following, we focus on solving the non-convex inner layer problem \eqref{EquInnLayProb} under given $\alpha_2 \in [0,1]$.

\subsection{Optimal Solution to Problem \eqref{EquInnLayProb} Via The Lagrange Dual Method}   \label{SecLagrangeTypeI}
First, we present the optimal solution to problem \eqref{EquInnLayProb}. Despite the non-convexity, problem \eqref{EquInnLayProb} can be shown to satisfy the ``time-sharing'' condition defined in \cite{Yu2006} as the number of sub-carriers $N$ tends to infinity, and the duality gap is zero in this case.\footnote{It is observed in our simulations that when $N = 32$, the duality gap for problem \eqref{EquInnLayProb} is negligibly small and thus can be ignored.} Hence, we apply the Lagrange dual method \cite{Boyd2004} to find its optimal solution.

The partial Lagrangian of problem \eqref{EquInnLayProb} is
\begin{align}
& L^{(\text{I})}(\mathbf{p}_{\text{PT}}, \mathbf{p}_{\text{IT}}, \mathbf{p}_{\text{J}}, \lambda, \mu)   \nonumber \\
= &  \sum_{n=1}^N \bigg[  \log_2 \left( 1 + \frac{ p_{\text{IT},n} | h_{\text{D},n} |^2 } {  p_{\text{J},n} | g_{\text{D},n} |^2 + \sigma_{\text{D}}^2 }  \right)  \nonumber \\
& -  \log_2 \left( 1 + \frac{ p_{\text{IT},n} | h_{\text{E},n} |^2 } {  p_{\text{J},n} | g_{\text{E},n} |^2 + \sigma_{\text{E}}^2 }  \right)  \bigg]   \nonumber   \\
& +\lambda  \bigg[ P_S -   (1- \alpha_2)  \sum_{n=1}^{N} p_{\text{PT},n} - \alpha_2 \sum_{n=1}^N p_{\text{IT},n}   \bigg]   \nonumber  \\
& + \mu \bigg[ (1-\alpha_2) \eta \sum_{n=1}^{N} p_{\text{PT},n} | h_{\text{J},n} |^2 - \alpha_2 \sum_{n=1}^N p_{\text{J},n}  \bigg],
\end{align}
where $\lambda \geq 0$ and $\mu \geq 0$ are the dual variables associated with the constraints \eqref{EquInnCon1} and \eqref{EquInnCon3}, respectively. The dual function is defined as
\begin{align}
g(\lambda,\mu) =  \max_{ \mathbf{p}_{\text{PT}}, \mathbf{p}_{\text{IT}}, \mathbf{p}_{\text{J}} } & L^{(\text{I})}(\mathbf{p}_{\text{PT}}, \mathbf{p}_{\text{IT}}, \mathbf{p}_{\text{J}}, \lambda, \mu)  \nonumber  \\
\text{s.t.} \quad & \; 0 \leq p_{\text{PT},n}  \leq P_{\text{S,peak}} , \; \forall n, \nonumber \\
& \; 0 \leq p_{\text{IT},n}  \leq P_{\text{S,peak}} , \; \forall n, \nonumber \\
& \;  0 \leq p_{\text{J},n}  \leq P_{\text{J,peak}} , \; \forall n. \label{EquDualFun}
\end{align}
Then, the dual problem of \eqref{EquInnLayProb} is
\begin{equation}   \label{EquDualProb}
\min_{\lambda,\mu}   \; g(\lambda,\mu) \;  \text{s.t.}   \; \lambda \geq 0, \; \mu \geq 0.  
\end{equation}
Due to the strong duality between problem \eqref{EquInnLayProb} and the dual problem \eqref{EquDualProb}, in the following we solve problem \eqref{EquInnLayProb} by first obtaining $g(\lambda,\mu)$ under given $\lambda \geq 0$ and $\mu \geq 0$ via solving problem \eqref{EquDualFun}, and then find the optimal $\lambda$ and $\mu$ to minimize $g(\lambda,\mu)$ for solving \eqref{EquDualProb}. 

First, consider problem \eqref{EquDualFun} under any given $\lambda \geq 0$ and $\mu \geq 0$. In this case, problem \eqref{EquDualFun} can be decomposed into $2N$ subproblems as follows by removing irrelevant terms, where each subproblems in \eqref{EquProbPPT} and \eqref{EquProbPITPJ} are for one sub-carrier $n$.
\begin{align}
\max_{p_{\text{PT},n}} & \; -\lambda (1-\alpha_2) p_{\text{PT},n} + \mu (1-\alpha_2 ) \eta | h_{\text{J},n} |^2 p_{\text{PT},n}    \nonumber  \\
\text{s.t.} \; & \; 0 \leq p_{\text{PT},n}  \leq P_{\text{S,peak}} ,  \label{EquProbPPT}
\end{align}
\begin{align}
\max_{p_{\text{IT},n},p_{\text{J},n}}  & \; \log_2 \left( 1 + \frac{ p_{\text{IT},n} | h_{\text{D},n} |^2 } {  p_{\text{J},n} | g_{\text{D},n} |^2 + \sigma_{\text{D}}^2 }  \right)  \nonumber \\
 -&   \log_2 \left( 1 + \frac{ p_{\text{IT},n} | h_{\text{E},n} |^2 } {  p_{\text{J},n} | g_{\text{E},n} |^2 + \sigma_{\text{E}}^2 }  \right) - \lambda \alpha_2 p_{\text{IT},n} - \mu \alpha_2  p_{\text{J},n} \nonumber \\
\text{s.t.} \; \;\;&  \; 0 \leq p_{\text{IT},n}  \leq P_{\text{S,peak}} ,   \nonumber \\
& \;  0 \leq p_{\text{J},n}  \leq P_{\text{J,peak}}.    \label{EquProbPITPJ}
\end{align}

As for subproblem \eqref{EquProbPPT}, as the objective function is linear over $p_{\text{PT},n}$, it is evident that the optimal solution is
\begin{equation}   \label{EquOptPPT}
p_{\text{PT},n}^* = 
\begin{cases}
P_{\text{S,peak}}, & - \lambda (1-\alpha_2) + \mu (1-\alpha_2) \eta  | h_{\text{J},n} |^2> 0, \\
0, & - \lambda (1-\alpha_2) + \mu (1-\alpha_2) \eta  | h_{\text{J},n} |^2 \leq 0.
\end{cases}
\end{equation}
Note that if $- \lambda (1-\alpha_2) + \mu (1-\alpha_2) \eta  | h_{\text{J},n} |^2=0$, $p_{\text{PT},n}^* $ is not unique, and can take any arbitrary value within $[0,P_{\text{S,peak}}]$. In this case, we set $p_{\text{PT},n}^* =0$ only for solving problem \eqref{EquDualFun}, which may not be the optimal solution of $p_{\text{PT},n}$ to problem \eqref{EquInnLayProb} in general.

As for subproblem \eqref{EquProbPITPJ}, the optimization variables $p_{\text{J},n}$ and $p_{\text{IT},n}$ couple together, thus making \eqref{EquProbPITPJ} difficult to solve. To handle this issue, we first obtain the optimal $p_{\text{IT},n}$ under any given $p_{\text{J},n} \in [0,P_{\text{J,peak}}]$,  and then apply a one-dimension search to find the optimal $p_{\text{J},n}$ within $[0,P_{\text{J,peak}}]$. To find the optimal $p_{\text{IT},n}$ to solve problem \eqref{EquProbPITPJ} under given $p_{\text{J},n}$, we define
\begin{equation}    \label{Equa}
a_n \triangleq \frac{  | h_{\text{D},n} |^2 } {   p_{\text{J},n} | g_{\text{D},n} |^2 + \sigma_{\text{D}}^2  },
\end{equation}
\begin{equation}   \label{Equb}
b_n \triangleq \frac{ | h_{\text{E},n} |^2 } {  p_{\text{J},n} | g_{\text{E},n} |^2 + \sigma_{\text{E}}^2  }.
\end{equation}
When $ a_n \leq b_n$, the objective function of \eqref{EquProbPITPJ} is non-increasing with respect to $p_{\text{IT},n}$, and the optimal solution of $p_{\text{IT},n}$ should be zero. When $a_n > b_n$, the objective function of \eqref{EquProbPITPJ} is concave with respect to $p_{\text{IT},n}$, and the optimal solution can be obtained by checking its first-order derivative. Therefore, the optimal $p_{\text{IT},n}$ for problem \eqref{EquProbPITPJ} under given $p_{\text{J},n}$ is
\begin{equation}    \label{EquOptPIT}
p_{\text{IT},n}^*(p_{\text{J},n}) =
\begin{cases}
0, & a_n \leq b_n,   \\
\min \left(  \left[ p_n^{*} \right]^{+}  , \; P_{\text{S,peak}} \right), & a_n > b_n,  
\end{cases}
\end{equation}
where
\begin{align}
p_n^{*} = &   \sqrt{ \left( \frac{1}{2 b_n} - \frac{1}{2 a_n} \right)^2  + \frac{1}{\lambda \alpha_2 \ln2} \left( \frac{1}{b_n} - \frac{1}{a_n} \right) }  \nonumber  \\
&- \frac{1}{2 b_n} - \frac{1}{2 a_n}.
\end{align}
In addition, let $p_{\text{J},n}^*$ denote the optimal $p_{\text{J},n}$ to problem \eqref{EquProbPITPJ}, obtained via the one-dimensional search. Then $p_{\text{IT},n}^*(p_{\text{J},n}^*)$ becomes the optimal solution of $p_{\text{IT},n}$ for \eqref{EquProbPITPJ}, denoted by $p_{\text{IT},n}^*$. By combining them with $p_{\text{PT},n}^*$ for \eqref{EquProbPPT}, the optimal solution to \eqref{EquDualFun} under given $(\lambda, \mu)$ is found.

Next, we solve the dual problem \eqref{EquDualProb}. As this problem is convex but may not be differentiable in general, we find the optimal $(\lambda, \mu)$ by applying the ellipsoid method \cite{Boyd2004}. The required subgradients of $g(\lambda,\mu)$ with respect to $\lambda$ and $\mu$ are respectively given by
\begin{equation}   \label{EquSubgradLambda}
P_S - (1-\alpha_2) \sum_{n=1}^N p_{\text{PT},n}^* - \alpha_2 \sum_{n=1}^N p_{\text{IT},n}^*,
\end{equation}
\begin{equation}   \label{EquSubgradMu}
(1-\alpha_2) \eta \sum_{n=1}^N p_{\text{PT},n}^* |h_{\text{J},n}|^2 - \alpha_2 \sum_{n=1}^N p_{\text{J},n}^*.
\end{equation}
Therefore, the optimal solution of \eqref{EquDualProb} can be obtained as $(\lambda^*,\mu^*)$.

With the optimal dual variable $(\lambda^*,\mu^*)$ at hand, the corresponding $p_{\text{IT},n}^*$'s and $p_{\text{J},n}^*$'s, which are obtained by solving problem \eqref{EquProbPITPJ}, become the optimal solution to problem \eqref{EquInnLayProb}. Now, it remains to obtain the optimal solution of $p_{\text{PT},n}$'s for problem \eqref{EquInnLayProb}. In general, the optimal solution of $p_{\text{PT},n}$'s, denoted as $p_{\text{PT},n}^*$'s, cannot be obtained from \eqref{EquOptPPT}, since the solution is not unique if $- \lambda^* (1-\alpha_2) + \mu^* (1-\alpha_2) \eta  | h_{\text{J},n} |^2 = 0$. Fortunately, it can be shown that, given $\lambda^*$, $\mu^*$, $p_{\text{IT},n}^*$'s, and $p_{\text{J},n}^*$'s, any $p_{\text{PT},n}$'s that satisfy the constraints \eqref{EquInnCon1}, \eqref{EquInnCon2}, and \eqref{EquInnCon3} are the optimal solution to problem \eqref{EquInnLayProb}. Thus we can find $p_{\text{PT},n}^*$'s by solving the following feasibility problem:
\begin{subequations}  \label{EquFeaPPT}
\begin{align}
\text{find} \; &  \mathbf{p}_{\text{PT}}     \\
\text{s.t.} \; &  (1- \alpha_2)  \sum_{n=1}^{N} p_{\text{PT},n} + \alpha_2 \sum_{n=1}^N p_{\text{IT},n}^* \leq P_S ,  \label{EquFeaPPTCon1}  \\
&0 \leq p_{\text{PT},n}  \leq P_{\text{S,peak}} ,\; n \in \mathcal{N}, \label{EquFeaPPTCon2}  \\
& \alpha_2 \sum_{n=1}^N p_{\text{J},n}^* \leq (1-\alpha_2) \eta \sum_{n=1}^{N} p_{\text{PT},n} | h_{\text{J},n} |^2.  \label{EquFeaPPTCon3}
\end{align}
\end{subequations}
The solution of problem \eqref{EquFeaPPT} can be obtained by solving the following problem.
\begin{align}
\max_{\mathbf{p}_{\text{PT}}} & \; \sum_{n=1}^{N} p_{\text{PT},n} | h_{\text{J},n} |^2  \label{EquFindOptPPT} \\
\text{s.t.} \; & \; \eqref{EquFeaPPTCon1}, \; \eqref{EquFeaPPTCon2}. \nonumber
\end{align}
This is because any solution to problem \eqref{EquFeaPPT} is a feasible solution to problem \eqref{EquFindOptPPT}, and thus the optimal solution to \eqref{EquFindOptPPT} must be a solution to problem \eqref{EquFeaPPT}. Let $\hat{k} = \lfloor (P_S - \alpha_2 \sum_{n=1}^N p_{\text{IT},n}^* ) / [ (1-\alpha_2)P_{\text{S,peak} } ] \rfloor$, where $\lfloor x \rfloor$ denotes the largest integer lower than $x$, and denote $|\tilde{h}_{\text{J},\hat{k}+1}|$ as the $(\hat{k}+1)$th largest value in $\{| h_{\text{J},n} |\}$. The optimal solution to problem \eqref{EquFindOptPPT} is
\begin{equation}    \label{EquOptPPTFinal}
p_{\text{PT},n}^* = 
\begin{cases}
P_{\text{S,peak}}, &  |h_{\text{J},n}| > |\tilde{h}_{\text{J},\hat{k}+1}|, \\
\frac{ P_S - \alpha_2 \sum_{n=1}^N p_{\text{IT},n}^* }{1-\alpha_2} - \hat{k} P_{\text{S,peak}},  &  |h_{\text{J},n}| = |\tilde{h}_{\text{J},\hat{k}+1}|, \\
0, &  |h_{\text{J},n}| < |\tilde{h}_{\text{J},\hat{k}+1}|.
\end{cases}
\end{equation}
Using \eqref{EquOptPPTFinal}, we obtain the closed-form optimal solution of $p_{\text{PT},n}$'s to problem \eqref{EquInnLayProb}.

In summary, the overall algorithm is presented in Algorithm \ref{AlgOptType1}. Denote the required accuracy for the one-dimension search in finding $p_{\text{J},n}$ and the convergence accuracy of the ellipsoid method as $\epsilon_{\text{J}}>0$ and $\epsilon_{\text{e}}>0$, respectively. The complexity of the Algorithm \ref{AlgOptType1} for finding the optimal solution is $\mathcal{O} \left[ N \left( \frac{P_{\text{J,peak}}}{\epsilon_{\text{J}}} + 1\right) \log_2 \frac{RG}{\epsilon_{\text{e}}} \right]$, where $R$ and $G$ are the radius and Lipschitz constant of the initial ellipsoid, respectively \cite{Boyd2014}.

\begin{algorithm}  
	\caption{The Optimal Solution to Problem \eqref{EquInnLayProb}}
	\begin{algorithmic}[1] \label{AlgOptType1}
		\STATE \textbf{Initialization:} Set an initial value of  $(\lambda,\mu)$ and an initial ellipsoid.
		\REPEAT
		\STATE Under given $(\lambda,\mu)$, for each $n$, obtain $p_{\text{PT},n}^*$'s by using \eqref{EquOptPPT}, and obtain $p_{\text{IT},n}^*$'s and $p_{\text{J},n}^*$'s by using \eqref{EquOptPIT} and a one-dimension search, respectively.
		\STATE Update $(\lambda,\mu)$ by using the ellipsoid method.
		\UNTIL {the volume of the ellipsoid is less than $\epsilon_{\text{e}}$.}
		\STATE Obtain $p_{\text{PT},n}^*$'s by using \eqref{EquOptPPTFinal}.
	\end{algorithmic}
\end{algorithm}

\subsection{Minorization Maximization (MM)}  \label{SectionMMType1}
Although the Lagrange dual method can find the optimal solution, it needs an exhaustive search of $p_{\text{J},n}$ to find the optimal power $p_{\text{J},n}^*$ and $p_{\text{IT},n}^*$ for each sub-carrier $n$. As a result, the computational complexity is rather high and even prohibitive for large $N$. Here, we propose a suboptimal approach to solve problem \eqref{EquInnLayProb} based on the MM approach \cite{Sun2017} to avoid exhaustive search, which obtains the power allocation solution iteratively. To facilitate the description, we rewrite \eqref{EquInnLayProb} as
\begin{align}
\max_{ \mathbf{p}_{\text{PT}}, \mathbf{p}_{\text{IT}}, \mathbf{p}_{\text{J}} } \; &  \sum_{n=1}^N \bigg[  \ln \left( p_{\text{IT},n} | h_{\text{D},n} |^2 + p_{\text{J},n}  | g_{\text{D},n} |^2 +  \sigma_{\text{D}}^2 \right)  \nonumber \\
& \; \; - \ln \left( p_{\text{J},n}  | g_{\text{D},n} |^2 +  \sigma_{\text{D}}^2 \right) \nonumber \\
& \; \; - \ln \left( p_{\text{IT},n} | h_{\text{E},n} |^2 + p_{\text{J},n}  | g_{\text{E},n} |^2 +  \sigma_{\text{E}}^2 \right)  \nonumber \\
& \; \; + \ln \left( p_{\text{J},n}  | g_{\text{E},n} |^2 +  \sigma_{\text{E}}^2 \right) \bigg]   \label{EquProDiffOfFourLog}  \\
\text{s.t.} \quad  &  \eqref{EquInnCon1} - \eqref {EquInnCon4},  \nonumber
\end{align}
where the property $\log_2 x = \ln x / \ln 2$ is used. The MM approach solves this problem iteratively as follows: in each iteration, this approach first constructs a surrogate function that is a concave lower bound of the objective function of the original problem, then maximizes the surrogate function within the feasible region of the original problem to obtain a feasible solution. The iteration terminates until the series of the obtained feasible solution converges. 

Without loss of generality, we consider the $(k+1)$-th iteration with $k \ge 0$. Suppose that $\mathbf{p}_{\text{PT}}^{(k)} = [ p_{\text{PT},1}^{(k)}, \ldots, p_{\text{PT},N}^{(k)} ]^\dagger$, $\mathbf{p}_{\text{IT}}^{(k)} = [ p_{\text{IT},1}^{(k)},  \ldots, p_{\text{IT},N}^{(k)} ]^\dagger$, $\mathbf{p}_{\text{J}}^{(k)} =[ p_{\text{J},1}^{(k)} , \ldots, p_{\text{J},N}^{(k)} ]^\dagger$ denote the solution obtained in the $k$-th iteration. We show how to find $\mathbf{p}_{\text{PT}}^{(k+1)}$, $\mathbf{p}_{\text{IT}}^{(k+1)}$ and $\mathbf{p}_{\text{J}}^{(k+1)}$ in the $(k+1)$-th iteration. Note that the first-order Taylor expansions of convex functions $- \ln ( p_{\text{J},n}  | g_{\text{D},n} |^2 +  \sigma_{\text{D}}^2 ) $ and $-  \ln ( p_{\text{IT},n} | h_{\text{E},n} |^2 + p_{\text{J},n}  | g_{\text{E},n} |^2 +  \sigma_{\text{E}}^2 )$ around $\mathbf{p}_{\text{IT}}^{(k)}$ and $\mathbf{p}_{\text{J}}^{(k)}$ are their respective global under-estimators \cite{Boyd2004}. Therefore, we have
\begin{equation}
\begin{split}
& - \ln \left( p_{\text{J},n}  | g_{\text{D},n} |^2 +  \sigma_{\text{D}}^2 \right) \\
\geq &- \frac{  | g_{\text{D},n} |^2 ( p_{\text{J},n} - p_{\text{J},n}^{(k)}  ) }{ p_{\text{J},n}^{(k)}  | g_{\text{D},n} |^2 +  \sigma_{\text{D}}^2 } - \ln \left( p_{\text{J},n}^{(k)}  | g_{\text{D},n} |^2 +  \sigma_{\text{D}}^2 \right),
\end{split}
\end{equation}
\begin{equation}
\begin{split}
& -  \ln \left( p_{\text{IT},n} | h_{\text{E},n} |^2 + p_{\text{J},n}  | g_{\text{E},n} |^2 +  \sigma_{\text{E}}^2 \right)  \\
\geq  &  -  \frac{ | h_{\text{E},n} |^2 ( p_{\text{IT},n} - p_{\text{IT},n}^{(k)} ) + | g_{\text{E},n} |^2 ( p_{\text{J},n} - p_{\text{J},n}^{(k)} )  }{ p_{\text{IT},n}^{(k)} | h_{\text{E},n} |^2 + p_{\text{J},n}^{(k)}  | g_{\text{E},n} |^2 +  \sigma_{\text{E}}^2 }   \\
& - \ln \left( p_{\text{IT},n}^{(k)} | h_{\text{E},n} |^2 + p_{\text{J},n}^{(k)}  | g_{\text{E},n} |^2 +  \sigma_{\text{E}}^2 \right).
\end{split}
\end{equation}

We construct a surrogate function of the objective function in \eqref{EquProDiffOfFourLog} by replacing $- \ln ( p_{\text{J},n}  | g_{\text{D},n} |^2 +  \sigma_{\text{D}}^2 ) $ and $-  \ln ( p_{\text{IT},n} | h_{\text{E},n} |^2 + p_{\text{J},n}  | g_{\text{E},n} |^2 +  \sigma_{\text{E}}^2 )$ with their respective first-order Taylor expansions. Then the maximization of the surrogate function within the feasible region of \eqref{EquProDiffOfFourLog} is expressed as
\begin{align}   
\max_{ \mathbf{p}_{\text{PT}}, \mathbf{p}_{\text{IT}}, \mathbf{p}_{\text{J}} } \; &  \sum_{n=1}^N \bigg[  \ln \left( p_{\text{IT},n} | h_{\text{D},n} |^2 + p_{\text{J},n}  | g_{\text{D},n} |^2 +  \sigma_{\text{D}}^2 \right) \nonumber \\
& \; \; + \ln \left( p_{\text{J},n}  | g_{\text{E},n} |^2 +  \sigma_{\text{E}}^2 \right) - \frac{  | g_{\text{D},n} |^2  p_{\text{J},n}   }{ p_{\text{J},n}^{(k)}  | g_{\text{D},n} |^2 +  \sigma_{\text{D}}^2 }  \nonumber \\
& \; \; - \frac{ | h_{\text{E},n} |^2  p_{\text{IT},n} + | g_{\text{E},n} |^2 p_{\text{J},n}  }{ p_{\text{IT},n}^{(k)} | h_{\text{E},n} |^2 + p_{\text{J},n}^{(k)}  | g_{\text{E},n} |^2 +  \sigma_{\text{E}}^2 } \bigg]  \label{EquMMIteProb}   \\
\text{s.t.} \quad   &   \eqref{EquInnCon1} - \eqref {EquInnCon4}, \nonumber
\end{align}
where the constant terms in the objective function are removed. Since the first and second summation terms in the objective function of \eqref{EquMMIteProb} are concave with respect to $p_{\text{IT},n}$ and $p_{\text{J},n}$, and the third and fourth summation terms in the objective function are linear, the objective function of \eqref{EquMMIteProb} is concave. Furthermore, the constraint functions in \eqref{EquInnCon1}--\eqref {EquInnCon4} are all convex, so the feasible region of \eqref{EquMMIteProb} is convex. As a result, problem \eqref{EquMMIteProb} is convex. We solve it by using the Lagrange dual method given in Appendix \ref{AppenDual}, without requiring the one-dimension exhaustive search applied in the optimal approach, and thus the complexity is lower.  

In summary, we have the MM approach as in Algorithm \ref{AlgMMType1}. Since problem \eqref{EquMMIteProb} maximizes the surrogate function which is a lower bound of the objective function of problem \eqref{EquInnLayProb}, and the lower bound and the objective function of (15) are equal only at the given point $(\mathbf{p}_{\text{PT}}^{(k)} , \mathbf{p}_{\text{IT}}^{(k)}, \mathbf{p}_{\text{J}}^{(k)} )$, the objective value of problem \eqref{EquInnLayProb} with the solution obtained by solving problem \eqref{EquMMIteProb} is non-decreasing over iteration. As the optimal value of \eqref{EquInnLayProb} is bounded from above, the MM approach is guaranteed to converge to at least a local optimum \cite{Sun2017}. The complexity of the MM approach is $\mathcal{O} \left[ N_{\text{Ite}} N \log_2 \frac{RG}{\epsilon_{\text{e}}}  \right]$, where $N_{\text{Ite}}$ is the iteration number.


\begin{algorithm}
	\caption{MM Approach to Solve Problem \eqref{EquInnLayProb}}
	\begin{algorithmic}[1]   \label{AlgMMType1}
		\STATE \textbf{Initialization:} Set an initial feasible solution $\mathbf{p}_{\text{PT}}^{(0)}$, $\mathbf{p}_{\text{IT}}^{(0)}$ and $\mathbf{p}_{\text{J}}^{(0)} $ and $k=0$.
		\REPEAT
		\STATE Set $k \gets k+1$;
		\STATE Solve problem \eqref{EquMMIteProb} by using the Lagrange dual method given in Appendix \ref{AppenDual} to find $\mathbf{p}_{\text{PT}}^{(k)}$, $\mathbf{p}_{\text{IT}}^{(k)}$ and $\mathbf{p}_{\text{J}}^{(k)} $.
		\UNTIL {The fractional increase of the objective value is below a small threshold $\epsilon_{\text{M}}$.}
	\end{algorithmic}
\end{algorithm}

\subsection{Heuristic Successive Optimization}  \label{SectionLowType1}
The previous two approaches are implemented iteratively and thus may have relatively high computation complexity. To overcome this issue, we further propose a low-complexity heuristic successive optimization by finding $\mathbf{p}_{\text{PT}}$, $\mathbf{p}_{\text{J}}$, and $\mathbf{p}_{\text{IT}}$ successively without any iteration. To this end, we decouple the variables $\mathbf{p}_{\text{PT}}$ and $\mathbf{p}_{\text{IT}}$ in the constraint \eqref{EquInnCon1}, and have the following problem: 
\begin{subequations}  \label{EquLowComProb}
	\begin{align}
	\max_{\mathbf{p}_{\text{PT}}, \mathbf{p}_{\text{IT}}, \mathbf{p}_{\text{J}} } \; &  \sum_{n=1}^N \bigg[  \ln \left( 1 + \frac{ p_{\text{IT},n} | h_{\text{D},n} |^2 } {  p_{\text{J},n} | g_{\text{D},n} |^2 + \sigma_{\text{D}}^2 }  \right) \nonumber \\
	& \; \; -  \ln \left( 1 + \frac{ p_{\text{IT},n} | h_{\text{E},n} |^2 } {  p_{\text{J},n} | g_{\text{E},n} |^2 + \sigma_{\text{E}}^2 }  \right)  \bigg]   \label{EquLowComObj}   \\
	\text{s.t.}  \quad  &  \sum_{n=1}^N p_{\text{PT},n} \leq P_S , \; 0 \leq p_{\text{PT},n} \leq P_{\text{S,peak}} , \forall n   \label{EquLowComCon1}  \\
	&  \sum_{n=1}^N p_{\text{IT},n} \leq P_S , \; 0 \leq p_{\text{IT},n} \leq P_{\text{S,peak}} , \forall n  \label{EquLowComCon2} \\
	&   \sum_{n=1}^N p_{\text{J},n} \leq \frac{1-\alpha_2}{\alpha_2} P_{\text{EH}} , \; 0 \leq p_{\text{J},n}  \leq P_{\text{J,peak}} , \forall n.  \label{EquLowComCon3}
	\end{align}
\end{subequations}
where $P_{\text{EH}} = \eta \sum_{n=1}^{N} p_{\text{PT},n} | h_{\text{J},n} |^2$ denotes the harvested power at the jammer. Problem \eqref{EquLowComProb} is obtained based on \eqref{EquInnLayProb} by replacing the constraints \eqref{EquInnCon1} and \eqref{EquInnCon2} with \eqref{EquLowComCon1} and \eqref{EquLowComCon2}. Since any variables $\mathbf{p}_{\text{PT}}$, $\mathbf{p}_{\text{IT}}$, and $\mathbf{p}_{\text{J}}$ satisfying \eqref{EquLowComCon1} and \eqref{EquLowComCon2} must satisfy \eqref{EquInnCon1} and \eqref{EquInnCon2}, the feasible region of problem \eqref{EquLowComProb} is a subset of that of \eqref{EquInnLayProb}. Therefore, solving \eqref{EquLowComProb} will result in a feasible solution to \eqref{EquInnLayProb} and achieve its lower bound. 

Next, we solve problem \eqref{EquLowComProb} by finding $\mathbf{p}_{\text{PT}}$, $\mathbf{p}_{\text{J}}$ and $\mathbf{p}_{\text{IT}}$ successively as follows.

\emph{1) Solution of} $\mathbf{p}_{\text{PT}}$.
Note that the optimal value of \eqref{EquLowComProb} can be viewed as a function of $P_{\text{EH}}$, denoted by $S(P_{\text{EH}})$. It is evident that for any given $P_{\text{EH}, 1} \geq P_{\text{EH},2} $, we have $S(P_{\text{EH},1}) \geq S(P_{\text{EH},2})$. This is due to the fact that the larger $P_{\text{EH}, 1}$ can admit a larger feasible region for $\mathbf{p}_{\text{PT}}$, $\mathbf{p}_{\text{IT}}$, and $\mathbf{p}_{\text{J}}$ for problem \eqref{EquLowComProb}, as compared to that admitted by $P_{\text{EH},2}$ (see \eqref{EquLowComCon3}). Therefore, $S(P_{\text{EH}})$ is non-decreasing function of $P_{\text{EH}}$. As a result, although $\mathbf{p}_{\text{PT}}$ is not directly involved in the objective function \eqref{EquLowComObj}, increasing $P_{\text{EH}}$ in \eqref{EquLowComCon3} can increase the objective value in \eqref{EquLowComObj}.

Hence, we propose to find the desirable $\mathbf{p}_{\text{PT}}$ by maximizing $P_{\text{EH}} = \sum_{n=1}^{N} p_{\text{PT},n} | h_{\text{J},n} |^2$. This corresponds to allocating power over the sub-carriers with highest channel gains as follows. Sort the sequence $\{ | h_{\text{J},n} | \}$ in the descent order and form a new sequence $\{| \tilde{h}_{\text{J},n} |\}$, where $| \tilde{h}_{\text{J},1} | \geq | \tilde{h}_{\text{J},2} | \geq \ldots \geq | \tilde{h}_{\text{J},N} |$. Let $k=\lfloor P_S / P_{\text{S,peak}} \rfloor$. Then we set 
\begin{equation}   \label{EquPPT}
p_{\text{PT},n} =
\begin{cases}
P_{\text{S,peak}} &  \text{if } | h_{\text{J},n} | \geq | \tilde{h}_{\text{J},k} |,  \\
P_S - kP_{\text{S,peak}} &  \text{if } | h_{\text{J},n} | = | \tilde{h}_{\text{J},k+1} |,  \\
0 & \text{otherwise.}
\end{cases}
\end{equation}

Consequently, the harvested power at the jammer is
\begin{equation}   \label{EquPEH}
P_{\text{EH}} = \eta \left[ P_{\text{S,peak}} \sum_{n=1}^{k} | \tilde{h}_{\text{J},n} |^2 + (P_S - kP_{\text{S,peak}}) | \tilde{h}_{\text{J},k+1} |^2 \right].
\end{equation}

\emph{2) Solution of} $\mathbf{p}_{\text{J}}$.
After obtaining $\mathbf{p}_{\text{PT}}$ and by substituting \eqref{EquPPT} into problem \eqref{EquLowComProb}, the optimization over $\mathbf{p}_{\text{IT}}$ and $\mathbf{p}_{\text{J}}$ becomes 
\begin{align}  
\max_{ \mathbf{p}_{\text{IT}}, \mathbf{p}_{\text{J}} } \; &  \sum_{n=1}^N \bigg[  \ln \left( 1 + \frac{ p_{\text{IT},n} | h_{\text{D},n} |^2 } {  p_{\text{J},n} | g_{\text{D},n} |^2 + \sigma_{\text{D}}^2 }  \right)  \nonumber \\
& \; \;-  \ln \left( 1 + \frac{ p_{\text{IT},n} | h_{\text{E},n} |^2 } {  p_{\text{J},n} | g_{\text{E},n} |^2 + \sigma_{\text{E}}^2 }  \right)  \bigg]  \label{EquProb3}  \\
\text{s.t.}  \; \;  &  \eqref{EquLowComCon2} ,  \nonumber \\
& \sum_{n=1}^N p_{\text{J},n} \leq  P_{\text{J,total}} ,  0 \leq p_{\text{J},n}  \leq P_{\text{J,peak}} , n \in \mathcal{N} ,  \nonumber
\end{align}
where $P_{\text{J,total}} = \frac{1-\alpha_2}{\alpha_2} P_{\text{EH}}$. As $\mathbf{p}_{\text{IT}}$ has not been obtained at this stage, in order to find $\mathbf{p}_{\text{J}}$, we adopt an equal power allocation over the sub-carriers when jamming is necessary. Considering the sub-carrier $n$, we consider jamming is necessary at that sub-carrier if increasing $p_{\text{J},n}$ at that sub-carrier will increase the objective function in \eqref{EquProb3}. Then, the jamming power is equally allocated over such necessary sub-carriers. We have the following lemma. 

\begin{lemma}  \label{LemmaSJ}
	If $| g_{\text{E},n} |^2/ \sigma_{\text{E}}^2 > | g_{\text{D},n} |^2 / \sigma_{\text{D}}^2$ for sub-carrier $n$, then jamming is necessary at that sub-carrier, i.e., increasing the jamming power at sub-carrier $n$ can increase the secrecy rate in the objective function of \eqref{EquProb3}. 
\end{lemma}

\begin{proof}
See Appendix \ref{Appen1}. 
\end{proof}

\emph{Remarks:} Note that $| g_{\text{E},n} |^2 / \sigma_{\text{E}}^2$ and $| g_{\text{D},n} |^2 /  \sigma_{\text{D}}^2$ are effective channel gains from the jammer to the eavesdropper and the destination, respectively. Lemma \ref{LemmaSJ} shows that in order to improve the secrecy rate of the system, jamming power should be allocated to the sub-carriers where the effective jamming channel gains to the eavesdropper are stronger than that to the destination.

Denote the set of sub-carriers over which jamming is necessary as
\begin{equation}   \label{EquSJ}
\mathcal{S}_{\text{J}} \triangleq \left \{ n \big| \frac{ | g_{\text{E},n} |^2 }{ \sigma_{\text{E}}^2 }   > \frac{ | g_{\text{D},n} |^2 } { \sigma_{\text{D}}^2 }  \right \}.
\end{equation}
Based on Lemma \ref{LemmaSJ}, we allocate the jamming power equally over the sub-carriers in $\mathcal{S}_{\text{J}}$, i.e.,
\begin{equation}   \label{EquPJLow}
p_{\text{J},n} =
\begin{cases}
\frac{P_{\text{J,total}}} {|\mathcal{S}_{\text{J}}|}, & n \in \mathcal{S}_{\text{J}},  \\
0, & \text{otherwise}.
\end{cases}
\end{equation}

\emph{3) Solution of} $\mathbf{p}_{\text{IT}}$.
For notational convenience, we define $a_n$ and $b_n$ as in \eqref{Equa} and \eqref{Equb}. By substituting \eqref{EquPJLow}, problem \eqref{EquLowComProb} becomes
\begin{align}  
\max_{ \mathbf{p}_{\text{IT}} }  \; & \sum_{n =1}^{N} \left[ \ln(1 + a_n p_{\text{IT},n} ) - \ln(1 + b_n p_{\text{IT},n} ) \right]  \label{EqProbLowComPIT} \\
\text{s.t.} \; \; & \eqref{EquLowComCon2}.  \nonumber
\end{align}
When $ a_n \leq b_n$, the objective function of \eqref{EqProbLowComPIT} is non-increasing function of $p_{\text{IT},n}$, and the optimal solution should be $p_{\text{IT},n}=0$. When $a_n > b_n$, the objective function of \eqref{EqProbLowComPIT} is concave with respect to $p_{\text{IT},n}$, and the optimal solution can be obtained by taking derivative of the objective function of \eqref{EqProbLowComPIT} with respect to $p_{\text{IT},n}$ and setting it to zero. As a result, the optimal power allocation solution to problem \eqref{EqProbLowComPIT}, is given by
\begin{equation}   \label{EquPIT}
p_{\text{IT},n} =
\begin{cases}
\min \left(  \left[ \tilde{p}_n \right]^{+}  , \; P_{\text{S,peak}} \right), & n \in \mathcal{S}_{\text{IT}} ,  \\
0, & \text{otherwise},
\end{cases}
\end{equation}
where
\begin{equation}   \label{EquPITLowSol2}
\tilde{p}_n = - \frac{1}{2 b_n} - \frac{1}{2 a_n} + \sqrt{ \left( \frac{1}{2 b_n} - \frac{1}{2 a_n} \right)^2  + \frac{1}{\vartheta} \left( \frac{1}{b_n} - \frac{1}{a_n} \right) },
\end{equation}
\begin{equation}  \label{EquSIT}
\mathcal{S}_{\text{IT}} \triangleq \{ n | a_n > b_n \},
\end{equation}
and $\vartheta \in [0, \max_{n \in \mathcal{S}_{\text{IT}}} (a_n - b_n)]$ guarantees the power constraint \eqref{EquLowComCon2} to be satisfied with equality, and it can be determined by bisection search. The upper bound of $\vartheta$ is found in Appendix \ref{Appen2}. The heuristic successive optimization approach is summarized in Algorithm \ref{AlgHeuType1}. The complexity of it is $\mathcal{O}(N)$.

\begin{algorithm}
	\caption{Heuristic Successive Optimization for Problem \eqref{EquInnLayProb}}
	\begin{algorithmic}[1]  \label{AlgHeuType1}
		\STATE Obtain $\mathbf{p}_{\text{PT}}$ using \eqref{EquPPT}, calculate $P_{\text{EH}}$ according to \eqref{EquPEH}, $P_{\text{J,total}} = [(1-\alpha_2)/\alpha_2] P_{\text{EH}}$.
		\STATE Obtain $\mathbf{p}_{\text{J}}$ using \eqref{EquPJLow}.
		\STATE Obtain $\mathbf{p}_{\text{IT}}$ using \eqref{EquPIT} where $\vartheta$ is found by using a bisection search over $[0, \max_{n \in \mathcal{S}_{\text{IT}}} (a_n - b_n)]$.
	\end{algorithmic}
\end{algorithm}

\section{Solution to Problem (P2) with Type-II Destination Receiver}
Now, we consider problem (P2) with Type-II destination receiver. Similarly as for problem (P1) with Type-I receiver, we reformulate this problem in the following form with outer and inner layers. 
\begin{equation}  \label{EquOutLayProbType2}
\max_{\alpha_2} \;   \alpha_2 \mathcal{R}^{(\text{II})}(\alpha_2), \; \;   \text{s.t.} \;   0 \leq \alpha_2 \leq 1,
\end{equation}
where
\begin{align}  
	\mathcal{R}^{(\text{II})}(\alpha_2) = & \max_{ \mathbf{p}_{\text{PT}}, \mathbf{p}_{\text{IT}}, \mathbf{p}_{\text{J}} }  \sum_{n=1}^N \bigg[  \log_2 \left( 1 + \frac{ p_{\text{IT},n} | h_{\text{D},n} |^2 } { \sigma_{\text{D}}^2 }  \right) \nonumber \\
	& \; \; \quad -  \log_2 \left( 1 + \frac{ p_{\text{IT},n} | h_{\text{E},n} |^2 } {  p_{\text{J},n} | g_{\text{E},n} |^2 + \sigma_{\text{E}}^2 }  \right)  \bigg]    \label{EquInnLayProbType2}   \\
	& \; \quad \text{s.t.} \; \;  \eqref{EquInnCon1} - \eqref {EquInnCon4} . \nonumber
\end{align}
As problem \eqref{EquOutLayProbType2} can be solved by a one-dimensional search over the interval $[0,1]$, we only need to focus on solving problem \eqref{EquInnLayProbType2} under given time allocation $\alpha_2$. In the following, we proposed the optimal, suboptimal and heuristic approaches, respectively, similarly as in the previous section for problem \eqref{EquInnLayProb}.

\subsection{Optimal Solution to Problem \eqref{EquInnLayProbType2} Via The Lagrange Dual Method}
Similar to Section \ref{SecLagrangeTypeI}, we apply the Lagrange dual approach to obtain the optimal solution to problem \eqref{EquInnLayProbType2}. The partial Lagrangian of \eqref{EquInnLayProbType2} is
\begin{align}
& L^{\text{(II)}} (\mathbf{p}_{\text{PT}}, \mathbf{p}_{\text{IT}}, \mathbf{p}_{\text{J}}, \lambda, \mu )   \nonumber \\
= &  \sum_{n=1}^N \bigg[  \log_2 \left( 1 + \frac{ p_{\text{IT},n} | h_{\text{D},n} |^2 } { \sigma_{\text{D}}^2 }  \right) \nonumber \\
&  -  \log_2 \left( 1 + \frac{ p_{\text{IT},n} | h_{\text{E},n} |^2 } {  p_{\text{J},n} | g_{\text{E},n} |^2 + \sigma_{\text{E}}^2 }  \right)  \bigg]     \nonumber  \\
& +\lambda  \bigg[ P_S -   (1- \alpha_2)  \sum_{n=1}^{N} p_{\text{PT},n} - \alpha_2 \sum_{n=1}^N p_{\text{IT},n}   \bigg]   \nonumber  \\
& + \mu \bigg[ (1-\alpha_2) \eta \sum_{n=1}^{N} p_{\text{PT},n} | h_{\text{J},n} |^2 - \alpha_2 \sum_{n=1}^N p_{\text{J},n}  \bigg],
\end{align}
where $\lambda \geq 0$ and $\mu \geq 0$ are the dual variables associated with the constraints \eqref{EquInnCon1} and \eqref{EquInnCon3}. The dual function is defined as
\begin{align}
g(\lambda,\mu) =  \max_{ \mathbf{p}_{\text{PT}}, \mathbf{p}_{\text{IT}}, \mathbf{p}_{\text{J}} } & L^{(\text{II})}(\mathbf{p}_{\text{PT}}, \mathbf{p}_{\text{IT}}, \mathbf{p}_{\text{J}}, \lambda, \mu)  \nonumber  \\
\text{s.t.} \quad & \; 0 \leq p_{\text{PT},n}  \leq P_{\text{S,peak}} , \; \forall n, \nonumber \\
& \; 0 \leq p_{\text{IT},n}  \leq P_{\text{S,peak}} , \; \forall n, \nonumber \\
& \;  0 \leq p_{\text{J},n}  \leq P_{\text{J,peak}} , \; \forall n. \label{EquDualFunType2}
\end{align}
Then, the dual problem of \eqref{EquInnLayProbType2} is
\begin{equation}  \label{EquDualProbType2}
\min_{\lambda,\mu}  \; g(\lambda,\mu) \; \text{s.t.}  \; \lambda \geq 0, \; \mu \geq 0.  
\end{equation}
First, we solve problem \eqref{EquDualFunType2} under any given $\lambda \geq 0$ and $\mu \geq 0$, which can be decomposed into $2N$ subproblems as follows, each for one sub-carrier $n$.
\begin{align}
\max_{p_{\text{PT},n}} & \; -\lambda (1-\alpha_2) p_{\text{PT},n} + \mu (1-\alpha_2 ) \eta | h_{\text{J},n} |^2 p_{\text{PT},n}    \nonumber   \\
\text{s.t.} \; & \; 0 \leq p_{\text{PT},n}  \leq P_{\text{S,peak}} , \label{EquProbPPTType2}
\end{align}
\begin{align}
& \max_{p_{\text{IT},n},p_{\text{J},n}} \;  \log_2 \left( 1 + \frac{ p_{\text{IT},n} | h_{\text{D},n} |^2 } {   \sigma_{\text{D}}^2 }  \right)   \nonumber \\
& \quad \; \; \; -  \log_2 \left( 1 + \frac{ p_{\text{IT},n} | h_{\text{E},n} |^2 } {  p_{\text{J},n} | g_{\text{E},n} |^2 + \sigma_{\text{E}}^2 }  \right) - \lambda \alpha_2 p_{\text{IT},n}  - \mu \alpha_2  p_{\text{J},n}    \nonumber  \\
& \quad \;  \text{s.t.} \; \;\; \;   0 \leq p_{\text{IT},n}  \leq P_{\text{S,peak}} , \nonumber \\
& \quad \quad \; \; \; \; \; \; 0 \leq p_{\text{J},n}  \leq P_{\text{J,peak}}. \label{EquProbPITPJType2}
\end{align}
Subproblem \eqref{EquProbPPTType2} is the same with problem \eqref{EquProbPPT}, so the solution can be obtained by \eqref{EquOptPPT}. Subproblem \eqref{EquProbPITPJType2} can be solved by the same method of solving problem \eqref{EquProbPITPJ}. The optimal $p_{\text{IT},n}$ with given $p_{\text{J},n}$ can be obtained by \eqref{EquOptPIT}, provided that $a_n$ is revised to be
\begin{equation}    
a_n = \frac{  | h_{\text{D},n} |^2 } {    \sigma_{\text{D}}^2  }.
\end{equation}
The optimal $p_{\text{J},n}$ is obtained by a one-dimension search within $[0,P_{\text{J,peak}}]$.

To solve \eqref{EquDualProbType2}, the pair $(\lambda,\mu)$ is updated by the ellipsoid method \cite{Boyd2004}, and the subgradients for $\lambda$ and $\mu$ are the same with \eqref{EquSubgradLambda} and \eqref{EquSubgradMu}. With the optimal dual variables $\lambda^*$ and $\mu^*$, the corresponding optimal $p_{\text{IT},n}^*$'s and $p_{\text{J},n}^*$'s, which are obtained by by solving problem \eqref{EquProbPITPJType2}, become optimal to problem \eqref{EquInnLayProbType2}. The optimal $p_{\text{PT},n}^*$'s to problem \eqref{EquInnLayProbType2} can be obtained by using \eqref{EquOptPPTFinal}.

The overall algorithm is similar as Algorithm \ref{AlgOptType1}, and is thus omitted here for brevity. Its complexity is $\mathcal{O} \left[ N \left( \frac{P_{\text{J,peak}}}{\epsilon_{\text{J}}} + 1\right) \log_2 \frac{RG}{\epsilon_{\text{e}}} \right]$ \cite{Boyd2014}.

\subsection{Minorization Maximization (MM)}
For the same reason expressed in Section \ref{SectionMMType1}, we propose a suboptimal approach to solve problem \eqref{EquInnLayProbType2} based on the MM approach. We rewrite \eqref{EquInnLayProbType2} as
\begin{align}  
\max_{ \mathbf{p}_{\text{PT}}, \mathbf{p}_{\text{IT}}, \mathbf{p}_{\text{J}} }  &  \sum_{n=1}^N \bigg[  \ln \left( p_{\text{IT},n} | h_{\text{D},n} |^2 +  \sigma_{\text{D}}^2 \right) + \ln \left( p_{\text{J},n}  | g_{\text{E},n} |^2 +  \sigma_{\text{E}}^2 \right)  \nonumber \\
& \; \;- \ln \left( p_{\text{IT},n} | h_{\text{E},n} |^2 + p_{\text{J},n}  | g_{\text{E},n} |^2 +  \sigma_{\text{E}}^2 \right) \bigg] \label{EquProDiffOfLogType2} \\
\text{s.t.} \; \; \;  &  \eqref{EquInnCon1} - \eqref {EquInnCon4}. \nonumber
\end{align}
This subsection adopts the MM approach to solve problem \eqref{EquInnLayProbType2} by following a similar procedure as in Section \ref{SectionMMType1}. Denote $\mathbf{p}_{\text{PT}}^{(k)}$, $\mathbf{p}_{\text{IT}}^{(k)}$ and $\mathbf{p}_{\text{J}}^{(k)}$ as the solution in the $k$-th iteration. Next, in the $(k+1)$-th iteration, we construct the surrogate function of the objective function in \eqref{EquInnLayProbType2} by replacing $-\ln ( p_{\text{IT},n} | h_{\text{E},n} |^2 + p_{\text{J},n}  | g_{\text{E},n} |^2 +  \sigma_{\text{E}}^2 )$ as its first-order Taylor expansion around $\mathbf{p}_{\text{IT}}^{(k)}$ and $\mathbf{p}_{\text{J}}^{(k)}$, and then solve the following surrogate function maximization problem within the feasible region of \eqref{EquInnLayProbType2}. 
\begin{align}  
\max_{ \mathbf{p}_{\text{PT}}, \mathbf{p}_{\text{IT}}, \mathbf{p}_{\text{J}} } \; &  \sum_{n=1}^N \bigg[  \ln \left( p_{\text{IT},n} | h_{\text{D},n} |^2 +  \sigma_{\text{D}}^2 \right)   + \ln \left( p_{\text{J},n}  | g_{\text{E},n} |^2 +  \sigma_{\text{E}}^2 \right) \nonumber \\
& \quad \quad  - \frac{ p_{\text{IT},n} | h_{\text{E},n} |^2 + p_{\text{J},n} | g_{\text{E},n} |^2  }{ p_{\text{IT},n}^{(k)} | h_{\text{E},n} |^2 + p_{\text{J},n}^{(k)}  | g_{\text{E},n} |^2 +  \sigma_{\text{E}}^2 } \bigg]  \label{EquMMIteProbType2} \\
\text{s.t.} \quad  &    \eqref{EquInnCon1} - \eqref {EquInnCon4}.  \nonumber
\end{align}
Problem \eqref{EquMMIteProbType2} is convex and thus can be solved by the Lagrange dual method given in Appendix \ref{AppenDualType2}. We iterate this procedure until the obtained solution sequence converges. As a result, the MM based solution is found. The algorithm description is similar to Algorithm \ref{AlgMMType1}, and is omitted here for brevity. The complexity is $\mathcal{O} \left[ N_{\text{Ite}} N \log_2 \frac{RG}{\epsilon_{\text{e}}}  \right]$.

\subsection{Heuristic Successive Optimization}
In addition, we propose a non-iterative heuristic successive optimization with much lower implementation complexity. Similar as in Section \ref{SectionLowType1}, we obtain an efficient solution to problem \eqref{EquInnLayProbType2} by considering the following problem, where the constraints \eqref{EquLowComCon1} and \eqref{EquLowComCon2} replace the constraints \eqref{EquInnCon1} and \eqref{EquInnCon2} in \eqref{EquInnLayProbType2}.
\begin{align}  
	\max_{ \mathbf{p}_{\text{PT}}, \mathbf{p}_{\text{IT}}, \mathbf{p}_{\text{J}} } &   \sum_{n=1}^N \bigg[  \log_2 \left( 1 + \frac{ p_{\text{IT},n} | h_{\text{D},n} |^2 } { \sigma_{\text{D}}^2 }  \right)   \nonumber \\
	& \; \; -  \log_2 \left( 1 + \frac{ p_{\text{IT},n} | h_{\text{E},n} |^2 } {  p_{\text{J},n} | g_{\text{E},n} |^2 + \sigma_{\text{E}}^2 }  \right)  \bigg]   \label{EquLowComProbType2}  \\
	\text{s.t.} \; \; \; &   \eqref{EquLowComCon1}, \; \eqref{EquLowComCon2}, \; \eqref{EquInnCon3}, \;\eqref{EquInnCon4} .   \nonumber
\end{align}
In the following, we solve problem \eqref{EquLowComProbType2} by obtaining $\mathbf{p}_{\text{PT}}$, $\mathbf{p}_{\text{J}}$ and $\mathbf{p}_{\text{IT}}$ successively.

\emph{1) Solution of} $\mathbf{p}_{\text{PT}}$.
It is easy to show that the optimization over $\mathbf{p}_{\text{PT}}$ with Type-II destination receiver is indeed same as that with Type-I receiver in Section \ref{SectionLowType1}. Therefore, $\mathbf{p}_{\text{PT}}$ is obtained as in \eqref{EquPPT}.

\emph{2) Solution of} $\mathbf{p}_{\text{J}}$.
With $\mathbf{p}_{\text{PT}}$, the remaining optimization over $\mathbf{p}_{\text{IT}}$ and $\mathbf{p}_{\text{J}}$ is expressed as
\begin{align} 
 \max_{ \mathbf{p}_{\text{IT}}, \mathbf{p}_{\text{J}} } \; &   \sum_{n=1}^N \bigg[  \log_2 \left( 1 + \frac{ p_{\text{IT},n} | h_{\text{D},n} |^2 } { \sigma_{\text{D}}^2 }  \right)  \nonumber \\
 & -  \log_2 \left( 1 + \frac{ p_{\text{IT},n} | h_{\text{E},n} |^2 } {  p_{\text{J},n} | g_{\text{E},n} |^2 + \sigma_{\text{E}}^2 }  \right)  \bigg]   \label{EquLowComProbType2v2}  \\
 \text{s.t.} \;  \; &   \eqref{EquLowComCon2} ,  \nonumber \\
 & \sum_{n=1}^N p_{\text{J},n} \leq  P_{\text{J,total}} ,   0 \leq p_{\text{J},n}  \leq P_{\text{J,peak}} , n \in \mathcal{N}.   \nonumber
\end{align}
Similar to the case of Type-I receiver in Section \ref{SectionLowType1}, we obtain $\mathbf{p}_{\text{J}}$ by applying an equal power allocation over sub-carriers where jamming power is necessary to improve the secrecy rate. From \eqref{EquLowComProbType2v2}, it is observed that over all sub-carriers, setting $p_{\text{J},n}$ to be positive can increase the objective function of \eqref{EquLowComProbType2v2}. As a result, all sub-carriers should be jammed. Therefore, we have the equal jamming power allocation over all sub-carriers as
\begin{equation}  \label{EquPJType2Low}
p_{\text{J},n} = \frac{P_{\text{J,total}}}{N}, \; n \in \mathcal{N}.
\end{equation}

\emph{3) Solution of} $\mathbf{p}_{\text{IT}}$.
With $\mathbf{p}_{\text{PT}}$ and $\mathbf{p}_{\text{J}}$ obtained, the optimization over $\mathbf{p}_{\text{IT}}$ is expressed as the same form in \eqref{EqProbLowComPIT} , provided that $a_n$ and $b_n$ are revised to be
\begin{equation}
a_n = \frac{  | h_{\text{D},n} |^2 } {   \sigma_{\text{D}}^2  },
\end{equation}
\begin{equation}
b_n = \frac{ | h_{\text{E},n} |^2 } {  p_{\text{J},n} | g_{\text{E},n} |^2 + \sigma_{\text{E}}^2  }.
\end{equation}
As a result, \eqref{EquPIT} are directly applicable to obtain $\mathbf{p}_{\text{IT}}$.

By combining \eqref{EquPPT} for $\mathbf{p}_{\text{PT}}$, \eqref{EquPJType2Low} for $\mathbf{p}_{\text{J}}$, and \eqref{EquPIT} for $\mathbf{p}_{\text{IT}}$, a heuristic solution to problem \eqref{EquInnLayProbType2} is finally obtained. The algorithm description is similar to Algorithm \ref{AlgHeuType1}, and is omitted here for brevity. The complexity is $\mathcal{O}(N)$.

\section{Simulation results}
In this section, we conduct computer simulations to verify the performances of our proposed approaches, as compared to the following benchmark schemes under fixed time allocation $\alpha_1$ and $\alpha_2$, or without any jamming:
\begin{itemize}
	\item MM-based approach with fixed time allocation (abbreviated as ``MM w/ fixed TA'' ): This scheme fixes the time allocation $\alpha_2$ as a constant, under which the source and jammer cooperatively allocate their power allocations adaptively over sub-carriers to maximize the secrecy rate. Particularly, this corresponds to solving problems \eqref{EquInnLayProb} and \eqref{EquInnLayProbType2} by using the MM approach for Type-I and Type-II destination receivers, respectively. 
	\item Heuristic successive optimization with fixed time allocation (abbreviated as ``heuristic w/ fixed TA''): This scheme also fixes the time allocation $\alpha_2$ as a constant, and optimizes the transmit power allocation over sub-carriers as in Sections III-C and IV-C with Type-I and Type-II destination receivers, respectively.
	\item Conventional design without cooperative jamming (abbreviated as ``conventional w/o CJ''): This scheme does not employ any cooperative jamming by allocating all the time and power for WIT. Under both receiver types, the source optimizes its power allocation based on \eqref{EquPIT}, where $\mathcal{S}_{\text{IT}}=\mathcal{N}$ are set.
\end{itemize}

Note that the implementation of the MM approach for both Type-I and Type-II destination receivers depends on the initial power allocation solution $\mathbf{p}_{\text{PT}}^{(0)}$, $\mathbf{p}_{\text{IT}}^{(0)}$ and $\mathbf{p}_{\text{J}}^{(0)} $. In the simulations, they are chosen based on the heuristic designs as follows. First, the initial $\mathbf{p}_{\text{PT}}^{(0)}$ is obtained in \eqref{EquPPT}. Next, the initial $\mathbf{p}_{\text{J}}^{(0)}$ is obtained by \eqref{EquPJLow} for Type-I receiver and by \eqref{EquPJType2Low} for Type-II receiver. Finally, the initial $\mathbf{p}_{\text{IT}}^{(0)}$ is obtained by first finding the sub-carrier set $\mathcal{S}_{\text{IT}}$ according to \eqref{EquSIT} and then allocating the transmit power equally over all the sub-carriers in $\mathcal{S}_{\text{IT}}$.

In the simulations, the carrier frequency is 750MHz, and the bandwidth is 10MHz. The number of sub-carriers is set as $N=32$. The channel response vectors $\mathbf{h}_{\text{J}}$, $\mathbf{h}_{\text{D}}$, $\mathbf{h}_{\text{E}}$, $\mathbf{g}_{\text{D}}$, and $\mathbf{g}_{\text{E}}$ are assumed to be independent and identically distributed CSCG random variables with zero mean and propagation-distance-dependent variances. The variances of the elements of $\mathbf{h}_{\text{J}}$, $\mathbf{h}_{\text{D}}$, $\mathbf{h}_{\text{E}}$, $\mathbf{g}_{\text{D}}$, and $\mathbf{g}_{\text{E}}$ are $\zeta_0 (d_{\text{SJ}} / d_0)^{-\kappa}$, $\zeta_0 (d_{\text{SD}} /d_0)^{-\kappa}$, $\zeta_0 (d_{\text{SE}} /d_0)^{-\kappa}$, $\zeta_0 (d_{\text{JD}}/d_0)^{-\kappa}$, and $\zeta_0 (d_{\text{JE}}/d_0)^{-\kappa}$, respectively, where $d_{\text{SJ}}$, $d_{\text{SD}}$, $d_{\text{SE}}$, $d_{\text{JD}}$, and $d_{\text{JE}}$ denote the distances from the source to the jammer, from the source to the destination, from the source to the eavesdropper, from the jammer to the destination, from the jammer to the eavesdropper, respectively. Here, $\zeta_0 = -30$dB corresponds to the path loss at a reference distance of $d_0 = 1$m, and $\kappa = 3$ is the path-loss exponent. We assume that the destination and eavesdropper are located close to each other, and set $d_{\text{SD}}=d_{\text{SE}}=5$m. We also assume that the jammer is located on the straight line between the source and the destination (or the eavesdropper), so $d_{\text{JD}}=d_{\text{SD}}-d_{\text{SJ}}$ and $d_{\text{JE}}=d_{\text{SE}}-d_{\text{SJ}}$. The variances of additive Gaussian noises over each sub-carrier are $\sigma_{\text{D}}^2=\sigma_{\text{E}}^2 = \sigma^2 / N$, where $\sigma^2 = -100$dBm. The energy harvesting efficiency is set as $\eta=0.5$. When applying the ellipsoid method, the initial dual variables are set to $\lambda=100$ and $\mu=100$, and the initial ellipsoid is set as $(\lambda - 100)^2 + (\mu - 100)^2 \leq 20100$. The one-dimension search interval for finding $p_{\text{J},n}$ in Algorithm \ref{AlgOptType1} is set to $\epsilon_{\text{J}} = P_{\text{J,peak}}/1000$, the convergence accuracy of ellipsoid method is set to $\epsilon_{\text{e}}=10^{-4}$, and convergence threshold in Algorithm \ref{AlgMMType1} is set to $\epsilon_{\text{M}} = 10^{-4}$. Unless specified otherwise, the following simulation results are averaged over 500 random independent channel realizations.

\begin{figure}[!t]
	\centering
	\includegraphics[width=\columnwidth]{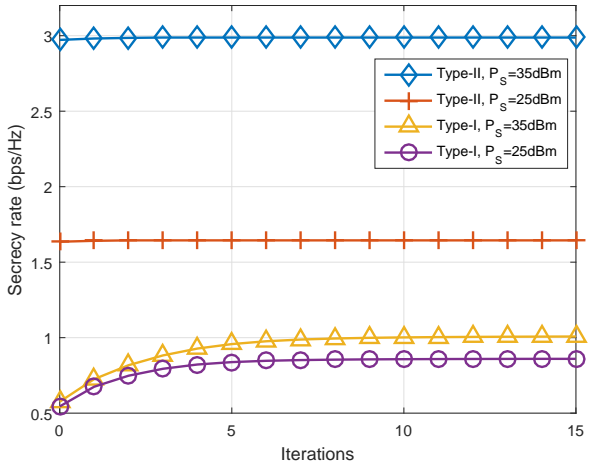}
	\caption{Secrecy rate of the MM approach vs. iteration number when $\alpha_2$=0.8 and $d_{\text{SJ}}=0.5$m.}   \label{FigRsIte}
\end{figure}

First, Fig. \ref{FigRsIte} shows the convergence behavior of the proposed MM approach for a given channel realization. The time portion is set as $\alpha_2=0.8$, and the distance from the source to the jammer is set as $d_{\text{SJ}}=0.5$m. The transmit power $P_S$ are 25dBm and 35dBm, respectively. In Fig. \ref{FigRsIte}, it is observed that the secrecy rates monotonically increase with the iteration number. With Type-I destination receiver, the secrecy rates converge within less than 10 iterations. With Type-II destination receiver, the secrecy rates converge within less than 5 iterations.

\begin{figure}[!t]
	\centering
	\includegraphics[width=\columnwidth]{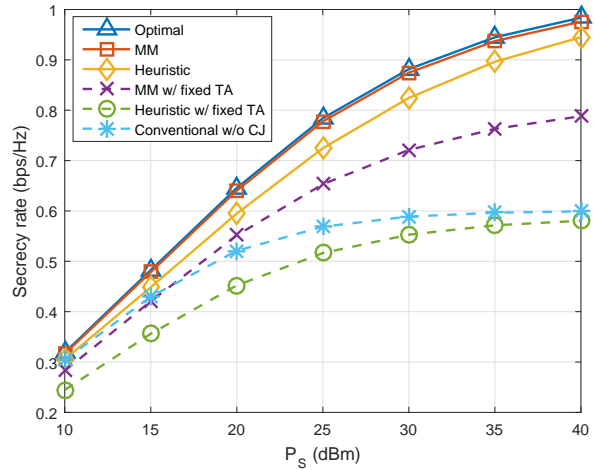}
	\caption{Secrecy rate vs. $P_S$ when $d_{\text{SJ}}=0.5$m (Type-I receiver at the destination)}   \label{FigRsPs}
\end{figure}

Next, Figs. \ref{FigRsPs} and \ref{FigRsPsType2} show the average secrecy rates versus the transmit power $P_S$ at the source, where the distance from the source to the jammer is set as $d_{\text{SJ}}=0.5$m. In Fig. \ref{FigRsPs} with Type-I destination receiver, it is observed that the average secrecy rates of all schemes increase as $P_S$ becomes large. The optimal Lagrange dual approach achieves the highest secrecy rate. The MM approach has a very close secrecy rate to the optimal approach. The heuristic successive optimization is observed to have a slightly lower secrecy rate than the optimal and MM approaches, but outperforms the other benchmark schemes significantly. This thus indicates the superiority of joint time and power allocation for improving secrecy rate, and validates the necessity of allocating time and power to wirelessly power the cooperative jamming in order to improve the secrecy rate of the OFDM communication.

\begin{figure}[!t]
	\centering
	\includegraphics[width=\columnwidth]{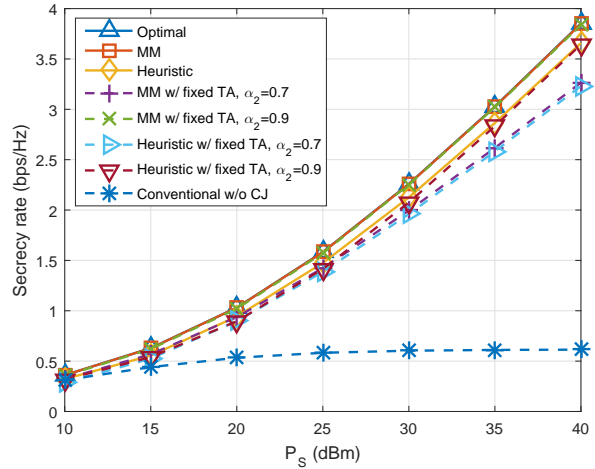}
	\caption{Secrecy rate vs. $P_S$ when $d_{\text{SJ}}=0.5$m (Type-II receiver at the destination)}   \label{FigRsPsType2}
\end{figure}

In Fig. \ref{FigRsPsType2} with Type-II destination receiver, it is observed that all schemes with jamming outperform the conventional one without cooperative jamming, which shows that wireless powered jamming is very effective in improving physical layer security. Similarly as in Fig. \ref{FigRsPs}, it is observed that the optimal approach achieves the highest secrecy rate performance. In addition, the performance gap between the MM approach and the optimal approach is very small. The heuristic successive optimization has slightly lower secrecy rate than the optimal and MM approaches but outperforms the others. Furthermore, it is observed that as compared with the Type-I destination receiver case in Fig. \ref{FigRsPs}, the secrecy rate under the Type-II destination receiver improves dramatically, thanks to the additional jamming signal cancellation at the Type-II receiver. 

\begin{figure}[!t]
	\centering
	\includegraphics[width=\columnwidth]{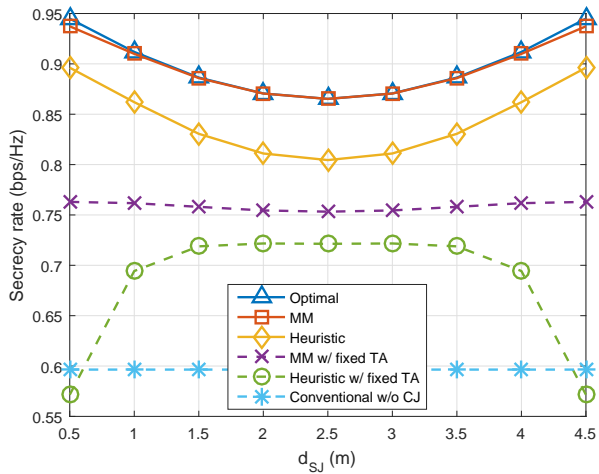}
	\caption{Secrecy rate vs. $d_{\text{SJ}}$ when $P_S=35$dBm (Type-I receiver at the destination)}
	\label{FigPsDsj}
\end{figure}

\begin{figure}[!t]
	\centering
	\includegraphics[width=\columnwidth]{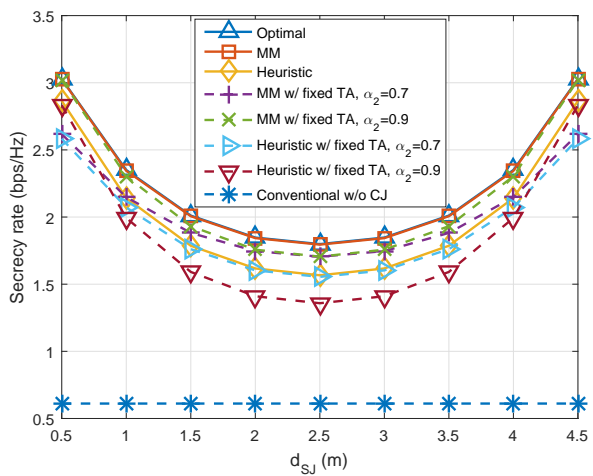}
	\caption{Secrecy rate vs. $d_{\text{SJ}}$ when $P_S=35$dBm (Type-II receiver at the destination)}
	\label{FigPsDsjType2}
\end{figure}

Then, we show the impact of node positions on the average secrecy rates of different schemes. We fix the positions of the source, destination and eavesdropper, and vary the position of the jammer. In the simulations, $d_{\text{SJ}}$ varies from 0.5m to 4.5m, and the power is set to be $P_S = 35$dBm. Fig. \ref{FigPsDsj} shows the secrecy rate versus $d_{\text{SJ}}$ with Type-I destination receiver. It is observed that when the jammer is moved from the source to the destination (with $d_{\text{SJ}}$ increasing from 0.5m to 4.5m), the average secrecy rates of the optimal, MM and heuristic successive optimization approaches first decrease and then increase, and the minimum secrecy rates of them are attained when the jammer is located at the middle between them, i.e., $d_{\text{SJ}}=2.5$m. It is also observed that when the jammer's location changes, the average secrecy rate achieved by the MM w/ fixed TA scheme only varies slightly, and that by the heuristic w/ fixed TA scheme first increases and then decreases.  The MM w/ fixed TA scheme and heuristic w/ fixed TA scheme are observed to outperform the conventional w/o CJ scheme significantly in most cases.

Fig. \ref{FigPsDsjType2} shows the secrecy rate versus $d_{\text{SJ}}$ with Type-II destination receiver. It is observed that when $d_{\text{SJ}}$ increases from 0.5m to 4.5m, the secrecy rates of all schemes with cooperative jamming first decrease and then increase, and the minimum secrecy rates of them are attained at $d_{\text{SJ}}=2.5$m. It can also observed that the conventional w/o CJ scheme always has the lowest secrecy rate.

Finally, Figs. \ref{FigAlpha2Dsj} and \ref{FigAlpha2DsjType2} show the optimal time portion $\alpha_2$ under the proposed three approaches versus $d_{\text{SJ}}$ with Type-I and Type-II destination receivers, respectively, where $P_S$ is fixed to 35dBm. It is observed that for all the three schemes, as $d_{\text{SJ}}$ increases, the optimal $\alpha_2$ first decreases and then increases, and the minimum is reached at $d_{\text{SJ}}=2.5$m. This means that when the jammer is located in the middle between the source and the destination, more time is allocated to the WPT time-slot to better utilize the cooperative jamming in this case. By contrast, when the jammer is located close to the source or eavesdropper, less time is allocated to the WPT time-slot. Furthermore, the optimal $\alpha_2$ of the optimal and MM approaches are almost the same, while the optimal $\alpha_2$ of the heuristic successive optimization is shorter. This means that the heuristic successive optimization needs longer WPT time. By comparing the Type-I and II receiver cases in Figs. \ref{FigAlpha2Dsj} and \ref{FigAlpha2DsjType2}, it is observed that the optimal $\alpha_2$ in the Type-II receiver case is lower than that in the Type-I receiver case. This is because the jamming signal cancellation ability of the Type-II receiver can fully use the effect of cooperative jamming, and thus the WPT for the jammer in the Type-II receiver case is allocated with more time resource. 

\begin{figure}[!t]
	\centering
	\includegraphics[width=\columnwidth]{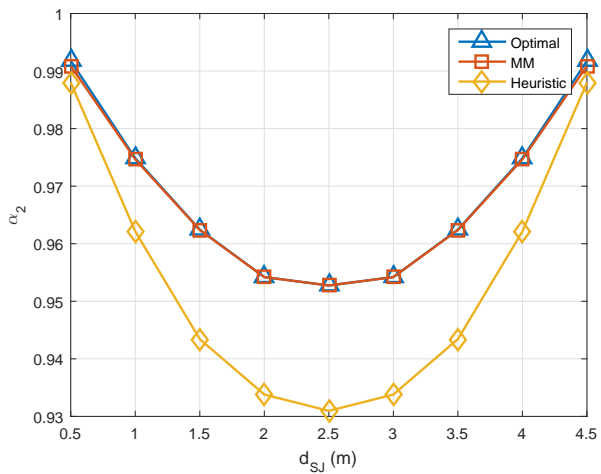}
	\caption{$\alpha_2$ vs. $d_{\text{SJ}}$ when $P_S=35$dBm (Type-I receiver at the destination)}
	\label{FigAlpha2Dsj}
\end{figure}

\begin{figure}[!t]
	\centering
	\includegraphics[width=\columnwidth]{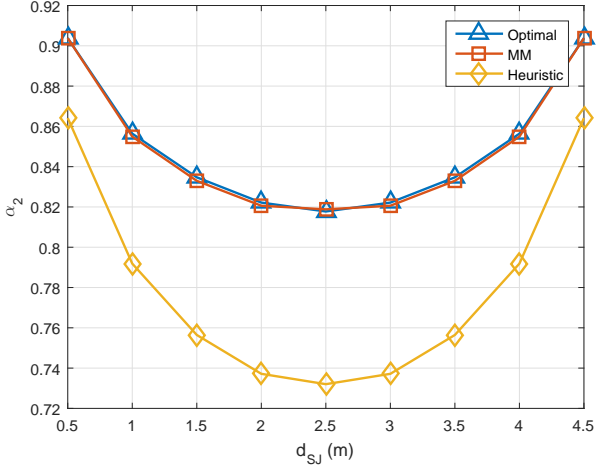}
	\caption{$\alpha_2$ vs. $d_{\text{SJ}}$ when $P_S=35$dBm (Type-II receiver at the destination)}
	\label{FigAlpha2DsjType2}
\end{figure}

\section{Conclusions}
We considered securing the OFDM communication between a source and a destination by exploiting wireless powered cooperative jamming, where a cooperative jammer harvests the wireless energy from the source and then uses the energy to jam the potential eavesdropper when the source is transmitting confidential information to the destination. We jointly design the time lengths and power allocated to WPT, WIT, and jamming to maximize the secrecy rate of the system, where two types of receivers deployed at the destination have been considered. A Lagrange dual approach and an MM approach have been proposed to find the optimal and suboptimal solutions of the joint design problem, and the secrecy rate performance gap between them is very small. A heuristic successive optimization has been further proposed to reduce the joint design complexity, which has slightly lower secrecy rate than the optimal and MM approaches. Simulation results show that joint time and power allocation is effective on improving the secrecy rate of OFDM communication systems.

\appendices

\section{Lagrange Dual Method for Problem \eqref{EquMMIteProb}}    \label{AppenDual}
Let
\begin{equation}
c_n = \frac{  | g_{\text{D},n} |^2 }{ p_{\text{J},n}^{(k)}  | g_{\text{D},n} |^2 +  \sigma_{\text{D}}^2 },
\end{equation}
\begin{equation}   \label{Equd}
d_n = \frac{ | h_{\text{E},n} |^2 }{ p_{\text{IT},n}^{(k)} | h_{\text{E},n} |^2 + p_{\text{J},n}^{(k)}  | g_{\text{E},n} |^2 +  \sigma_{\text{E}}^2 },
\end{equation}
\begin{equation}   \label{Eque}
e_n = \frac{ | g_{\text{E},n} |^2  }{ p_{\text{IT},n}^{(k)} | h_{\text{E},n} |^2 + p_{\text{J},n}^{(k)}  | g_{\text{E},n} |^2 +  \sigma_{\text{E}}^2 }.
\end{equation}
The partial Lagrangian of problem \eqref{EquMMIteProb} is
\begin{align}
& L_{\text{MM}}^{\text{(I)}}  (\mathbf{p}_{\text{PT}}, \mathbf{p}_{\text{IT}}, \mathbf{p}_{\text{J}}, \lambda, \mu)   \nonumber  \\
= &  \sum_{n=1}^N \big[  \ln \left( p_{\text{IT},n} | h_{\text{D},n} |^2 + p_{\text{J},n}  | g_{\text{D},n} |^2 +  \sigma_{\text{D}}^2 \right) \nonumber \\
& + \ln \left( p_{\text{J},n}  | g_{\text{E},n} |^2 +  \sigma_{\text{E}}^2 \right) - c_n p_{\text{J},n} - d_n p_{\text{IT},n} - e_n p_{\text{J},n} \big] \nonumber \\
&  +\lambda  \bigg[ P_S -   (1- \alpha_2)  \sum_{n=1}^{N} p_{\text{PT},n} - \alpha_2 \sum_{n=1}^N p_{\text{IT},n}   \bigg]   \nonumber  \\
& + \mu \bigg[ (1-\alpha_2) \eta \sum_{n=1}^{N} p_{\text{PT},n} | h_{\text{J},n} |^2 - \alpha_2 \sum_{n=1}^N p_{\text{J},n}  \bigg],  \label{EquLagMM1}
\end{align}
where $\lambda \geq 0$ and $\mu \geq 0$ are the dual variables associated with the constraints \eqref{EquInnCon1} and \eqref{EquInnCon3}. 
The dual function is
\begin{subequations}   \label{EquDualFunMM}
\begin{align}
g(\lambda,\mu) = &  \max_{ \mathbf{p}_{\text{PT}}, \mathbf{p}_{\text{IT}}, \mathbf{p}_{\text{J}} } \; L_{\text{MM}}^{\text{(I)}}  (\mathbf{p}_{\text{PT}}, \mathbf{p}_{\text{IT}}, \mathbf{p}_{\text{J}}, \lambda, \mu)  \\
& \quad \; \text{s.t.} \; \quad 0 \leq p_{\text{PT},n}  \leq P_{\text{S,peak}} , \; \forall n,  \label{EquPPTCon} \\
& \quad  \quad \; \; \;  \quad 0 \leq p_{\text{IT},n}  \leq P_{\text{S,peak}} , \; \forall n, \label{EquPITCon} \\
& \quad \quad \; \; \; \quad 0 \leq p_{\text{J},n}  \leq P_{\text{J,peak}} , \; \forall n. \label{EquPJCon}
\end{align}
\end{subequations}
The dual problem of problem \eqref{EquMMIteProb} is
\begin{equation}  \label{EquDualProbMM}
\min_{\lambda, \mu }  \;  g(\lambda,\mu) \; \text{s.t.} \; \lambda \geq 0, \; \mu \geq 0.  \\
\end{equation} 

For given $(\lambda,\mu)$, the optimal $p_{\text{PT},n}$ to problem \eqref{EquDualFunMM} is given in \eqref{EquOptPPT}. The optimal $p_{\text{IT},n}$ and $p_{\text{J},n}$ to problem \eqref{EquDualFunMM} can be obtained in closed-form as follows. Take derivative of the objective function of \eqref{EquDualFunMM} with respect to $p_{\text{IT},n}$ and set it to zero, we can find the relation between $p_{\text{IT},n}$ and $p_{\text{J},n}$:
\begin{equation}   \label{EquPITPJ}
p_{\text{IT},n} = -\frac{ |g_{\text{D},n}|^2 }{ |h_{\text{D},n}|^2 } p_{\text{J},n} - \frac{ \sigma_{\text{D}}^2 }{ |h_{\text{D},n}|^2 } + \frac{1}{d_n + \lambda \alpha_2}.
\end{equation}
Take derivative of the objective function of \eqref{EquDualFunMM} with respect to $p_{\text{J},n}$, we have
\begin{align}
\frac{\partial L_{\text{MM}}^{\text{(I)}}}{\partial p_{\text{J},n}} = & \frac{  | g_{\text{D},n} |^2 }{ p_{\text{IT},n} | h_{\text{D},n} |^2 + p_{\text{J},n}  | g_{\text{D},n} |^2 +  \sigma_{\text{D}}^2 } + \frac{ | g_{\text{E},n} |^2}{p_{\text{J},n}  | g_{\text{E},n} |^2 +  \sigma_{\text{E}}^2}   \nonumber \\
& -c_n - e_n - \mu \alpha_2.  \label{EquDLPJ}
\end{align}
We substitute \eqref{EquPITPJ} to \eqref{EquDLPJ} and set it to zero, and then find the solution of $p_{\text{J},n}$ as
\begin{equation}  \label{EquMMPJ}
\tilde{p}_{\text{J},n} = \frac{ 1 }{ c_n + e_n + \mu \alpha_2 - \frac{ |g_{\text{D},n}|^2 }{ |h_{\text{D},n}|^2 } (d_n + \lambda \alpha_2)  } - \frac{\sigma_{\text{E}}^2}{ |g_{\text{E},n}|^2 }. 
\end{equation}
Substituting \eqref{EquMMPJ} to \eqref{EquPITPJ}, we find the solution of $p_{\text{IT},n}$ as
\begin{align}   \label{EquMMPIT}
\tilde{p}_{\text{IT},n} = & \frac{1}{ d_n + \lambda \alpha_2 - \frac{ |h_{\text{D},n}|^2 }{| g_{\text{D},n} |^2} ( c_n + e_n + \mu \alpha_2 ) } + \frac{ | g_{\text{D},n} |^2 \sigma_{\text{E}}^2 }{ | h_{\text{D},n} |^2 | g_{\text{E},n} |^2 }      \nonumber \\
&   + \frac{1}{d_n + \lambda \alpha_2}  - \frac{\sigma_{\text{D}}^2}{|h_{\text{D},n}|^2}.
\end{align} 
Thus, the optimal solution of $p_{\text{IT},n}$ and $p_{\text{J},n}$ to problem \eqref{EquDualFunMM} is 
\begin{equation}  \label{EquMMOptPJ}
p_{\text{J},n}^* = \min \left( [ \tilde{p}_{\text{J},n} ]^+ , P_{\text{J,peak}} \right),
\end{equation}
\begin{equation} \label{EquMMOptPIT}
p_{\text{IT},n}^* = \min \left( [ \tilde{p}_{\text{IT},n} ]^+ , P_{\text{S,peak}} \right).
\end{equation}
Equ. \eqref{EquMMPJ}--\eqref{EquMMOptPIT} show that the solutions of $p_{\text{J},n}$'s and $p_{\text{IT},n}$'s to problem \eqref{EquDualFunMM} follow a water-filling structure with different water-levels across different sub-carriers.

To solve problem \eqref{EquDualProbMM}, the pair $(\lambda,\mu)$ can be updated by applying the ellipsoid method \cite{Boyd2004}. The required subgradients for updating $\lambda$ and $\mu$ are given by \eqref{EquSubgradLambda} and \eqref{EquSubgradMu}. With the optimal solution to \eqref{EquDualProbMM}, denoted as $\lambda^*$ and $\mu^*$, the optimal $p_{\text{J},n}$'s and $p_{\text{IT},n}$'s corresponding to $\lambda^*$ and $\mu^*$ are the optimal solution to problem \eqref{EquMMIteProb}. Then the optimal $p_{\text{PT},n}$'s to problem \eqref{EquMMIteProb} can be obtained by \eqref{EquOptPPTFinal}. The description of the overall method is similar as Algorithm \ref{AlgOptType1} and is thus omitted here for brevity.

\section{Proof of Lemma \ref{LemmaSJ}}  \label{Appen1}
First, we have the following fact that for arbitrary $c,d >0$, $(1+c)/(1+d)$ increases with $c/d$, which is proved as follows. Without loss of generality, we can increase $c/d$ by fixing $d$ and increasing $c$. It is observed that $(1+c)/(1+d)$ will also increase with $c$, when $d$ is fixed. So $(1+c)/(1+d)$ increases with $c/d$.
	
Then, we write the objective function of \eqref{EquProb3} into the following form
\begin{equation}  \label{EquProb3Obj}
	\ln \frac{ 1 + \left( p_{\text{IT},n} | h_{\text{D},n} |^2 \right) / \left( p_{\text{J},n}  | g_{\text{D},n} |^2 +  \sigma_{\text{D}}^2    \right)  }  { 1 + \left( p_{\text{IT},n} | h_{\text{E},n} |^2  \right) / \left( p_{\text{J},n}  | g_{\text{E},n} |^2 +   \sigma_{\text{E}}^2  \right) }.
\end{equation}
The previously proved fact tells that \eqref{EquProb3Obj} will increase with
\begin{equation}  \label{EquProb3ObjPart}
\begin{split}
	&\frac{  \left( p_{\text{IT},n} | h_{\text{D},n} |^2 \right) / \left( p_{\text{J},n}  | g_{\text{D},n} |^2 +  \sigma_{\text{D}}^2    \right)  }  {  \left( p_{\text{IT},n} | h_{\text{E},n} |^2  \right) / \left( p_{\text{J},n}  | g_{\text{E},n} |^2 +   \sigma_{\text{E}}^2  \right) } \\
	=& \frac{ | h_{\text{D},n} |^2 | g_{\text{E},n} |^2 } { | h_{\text{E},n} |^2 | g_{\text{D},n} |^2 } \left( 1 + \frac{ \frac{ \sigma_{\text{E}}^2 }{ | g_{\text{E},n} |^2 } - \frac{ \sigma_{\text{D}}^2 }{ | g_{\text{D},n} |^2 } } { p_{\text{J},n} + \frac{ \sigma_{\text{D}}^2 }{ | g_{\text{D},n} |^2 } }  \right).
\end{split}	
\end{equation}
Note that \eqref{EquProb3ObjPart} only increases with $p_{\text{J},n}$ when $ \sigma_{\text{E}}^2 /  | g_{\text{E},n} |^2 - \sigma_{\text{D}}^2 / | g_{\text{D},n} |^2 < 0$. Hence, increasing $p_{\text{J},n}$ can increase the objective function of \eqref{EquProb3} when $ | g_{\text{E},n} |^2 / \sigma_{\text{E}}^2  > | g_{\text{D},n} |^2 / \sigma_{\text{D}}^2 $.

\section{The Upper Bound of $\vartheta$}   \label{Appen2}
By taking derivative of the Lagrangian of problem \eqref{EqProbLowComPIT} and setting it to zero, we have the following equation 
\begin{equation}   \label{EquQuaEquation}
a_n b_n p_{\text{IT},n}^2 + (a_n+b_n) p_{\text{IT},n} - \frac{a_n-b_n}{\vartheta } +1 =0,
\end{equation}
where $\vartheta \in [0,\vartheta_{\max}]$ is the dual variable associated with the sum power constraint $\sum_{n=1}^N p_{\text{IT},n} \leq P_S$. The value of $\vartheta_{\max}$ can be obtained as follows. 

The $\tilde{p}_n$ in \eqref{EquPITLowSol2} is the positive root of \eqref{EquQuaEquation}. Since $\tilde{p}_n$ is a decreasing function of $\vartheta$, a sufficient large $\vartheta$ can make $\tilde{p}_n$ negative for all $n \in \mathcal{S}_{\text{IT}}$. According to \eqref{EquPIT}, negative $\tilde{p}_n$ makes $p_{\text{IT},n}=0$, which is obviously not the optimal solution of \eqref{EqProbLowComPIT}. Hence $\vartheta$ should be bounded above to make sure at least one $\tilde{p}_n$, $n \in \mathcal{S}_{\text{IT}}$ is positive. 

Note that the sum of the roots of equation \eqref{EquQuaEquation}, i.e. $-(a_n + b_n)/(a_n b_n)$, is negative, the condition that the equation has positive root is equivalent to the condition that the product of the roots is negative, i.e.
\begin{equation}
\frac{ - \frac{a_n-b_n}{\vartheta } +1 } { a_n b_n } < 0 \; \Rightarrow \; \vartheta  < a_n - b_n, \; \; n \in \mathcal{S}_{\text{IT}}.
\end{equation}
So 
\begin{equation}
\vartheta_{\max} = \max_{n \in \mathcal{S}_{\text{IT}}} (a_n - b_n).
\end{equation}

\section{Lagrange Dual Method for Problem \eqref{EquMMIteProbType2}}     \label{AppenDualType2}
The procedure of the method is the similar with that shown in Appendix \ref{AppenDual}, and the only differences are the expressions of the Lagrangian, the dual function, and the solution of $p_{\text{IT},n}$ and $p_{\text{J},n}$. We only show the differences here. 

The partial Lagrangian of problem \eqref{EquMMIteProbType2} is
\begin{align}
& L_{\text{MM}}^{\text{(II)}}  (\mathbf{p}_{\text{PT}}, \mathbf{p}_{\text{IT}}, \mathbf{p}_{\text{J}}, \lambda, \mu)   \nonumber  \\
= &  \sum_{n=1}^N \big[  \ln \left( p_{\text{IT},n} | h_{\text{D},n} |^2 +  \sigma_{\text{D}}^2 \right) + \ln \left( p_{\text{J},n}  | g_{\text{E},n} |^2 +  \sigma_{\text{E}}^2 \right) \nonumber \\
&   - d_n p_{\text{IT},n} - e_n p_{\text{J},n} \big] \nonumber \\
&  +\lambda  \bigg[ P_S -   (1- \alpha_2)  \sum_{n=1}^{N} p_{\text{PT},n} - \alpha_2 \sum_{n=1}^N p_{\text{IT},n}   \bigg]   \nonumber  \\
& + \mu \bigg[ (1-\alpha_2) \eta \sum_{n=1}^{N} p_{\text{PT},n} | h_{\text{J},n} |^2 - \alpha_2 \sum_{n=1}^N p_{\text{J},n}  \bigg],  \label{EquLagMM2}
\end{align}
where $d_n$ and $e_n$ are defined as \eqref{Equd} and \eqref{Eque}, and $\lambda \geq 0$ and $\mu \geq 0$ are the dual variables associated with constraints \eqref{EquInnCon1} and \eqref{EquInnCon3}. The dual function is
\begin{align}
	g(\lambda,\mu) = &  \max_{ \mathbf{p}_{\text{PT}}, \mathbf{p}_{\text{IT}}, \mathbf{p}_{\text{J}} } \; L_{\text{MM}}^{\text{(II)}}  (\mathbf{p}_{\text{PT}}, \mathbf{p}_{\text{IT}}, \mathbf{p}_{\text{J}}, \lambda, \mu)  \label{EquDualFunMMType2}  \\
	& \quad \; \text{s.t.} \; \; \quad \eqref{EquPPTCon} ,  \; \eqref{EquPITCon} , \; \eqref{EquPJCon}.  \nonumber
\end{align}

By taking derivatives of the objective value of problem \eqref{EquDualFunMMType2} with respect to $p_{\text{IT},n}$ and $p_{\text{J},n}$, respectively, and setting them to zero, we can find the optimal $p_{\text{IT},n}$ and $p_{\text{J},n}$ to problem \eqref{EquDualFunMMType2} as
\begin{equation}   \label{EquPITMMType2}
p_{\text{IT},n}^* = \min \left( [ \hat{p}_{\text{IT},n} ]^+ , P_{\text{S,peak}} \right),
\end{equation}
\begin{equation}  
p_{\text{J},n}^* = \min \left( [ \hat{p}_{\text{J},n} ]^+ , P_{\text{J,peak}} \right),
\end{equation}
where
\begin{equation}
\hat{p}_{\text{IT},n} = \frac{1}{d_n + \lambda \alpha_2} - \frac{\sigma_{\text{D}}^2}{|h_{\text{D},n}|^2},
\end{equation}
\begin{equation}  \label{EquPJMMType2}
\hat{p}_{\text{J},n} = \frac{1}{e_n + \mu \alpha_2} - \frac{\sigma_{\text{E}}^2}{|g_{\text{E},n}|^2}.
\end{equation}
Equ. \eqref{EquPITMMType2}--\eqref{EquPJMMType2} show that the solutions of $p_{\text{J},n}$'s and $p_{\text{IT},n}$'s to problem \eqref{EquDualFunMMType2} follow a water-filling structure with different water-levels across different sub-carriers.

\end{document}